\documentclass[11pt,letterpaper]{article}
\usepackage{hyperref}

\usepackage{lmodern}
\usepackage[T1]{fontenc}
\usepackage[utf8]{inputenc}

\usepackage{verbatim}
\usepackage{amsfonts}
\usepackage{amsmath,amssymb,amsthm}
\usepackage{url}
\usepackage{setspace}
\usepackage{times}
\usepackage{epsfig}
\usepackage[usenames]{color}
\usepackage{graphicx}
\usepackage{enumerate}

\def\yesd{\mathcal{YES}} \def\nod{\mathcal{NO}}
\newtheorem{theorem}{Theorem}[section]
\newtheorem{lemma}[theorem]{Lemma}
\newtheorem{property}[theorem]{Property}

\newtheorem{corollary}[theorem]{Corollary}

\newtheorem{claim}[theorem]{Claim}
\newtheorem{observation}[theorem]{Observation}
\newtheorem{definition}[theorem]{Definition}

\newtheorem{remark}[theorem]{Remark}

\newcommand{\overbar}[1]{\mkern 1.5mu\overline{\mkern-1.2mu#1\mkern-1.2mu}\mkern 1.5mu}
\usepackage[left=1in,right=1in,top=0.7in, bottom=0.8in]{geometry}

\newcommand{\wt}[1]{\mathrm{wt}(#1)}

   \def\11{\mathbf{1}}

\def\calP{\mathcal{P}}   
  \def\calD{\mathcal{D}} 
\def\calS{\mathcal{S}}  \def\inwt{\text{in-wt}} \def\calA{\mathcal{A}}
  \def\calB{\mathcal{B}}
\spacing{1.1}
   
\def\zero{\textsc{Zero}}  \def\dist{\text{dist}}
\def\fC{\frak C} \def\nil{\textsf{nil}}
\def\mconj{\textsc{Mconj}}
\def\conj{\textsc{Conj}}
\def\dl{\textsc{Dlist}}
\def\ltf{\textsc{LTF}}

\begin{document}

\title{Tight Bounds for the Distribution-Free Testing\\ of Monotone Conjunctions\vspace{0.3cm}}
\author{Xi Chen\\ Columbia University\\ \texttt{xichen@cs.columbia.edu}
 \and
Jinyu Xie\\ Columbia University\\ \texttt{jinyu@cs.columbia.edu}}
\date{}
\setcounter{page}{0}\maketitle
\thispagestyle{empty}

\begin{abstract} We improve both upper and lower bounds 
  for the distribution-free testing of monotone conjunctions.
Given oracle access to an unknown Boolean function $f\hspace{-0.02cm}:\hspace{-0.02cm}\{0,1\}^n
\hspace{-0.02cm}\rightarrow\hspace{-0.02cm} \{0,1\}$
  and sampling oracle access to an unknown distribution $\calD$ over $\{0,1\}^n$,
we present an $\tilde{O}(n^{1/3}/\epsilon^5)$-query algorithm that tests
  whether $f$ is a monotone conjunction versus $\epsilon$-far from any monotone conjunction
  with respect to $\calD$.
This improves the previous best upper bound of $\tilde{O}(n^{1/2}/\epsilon)$ by Dolev and Ron \cite{DolevRon}
  when $1/\epsilon$ is small compared to $n$.
For some constant $\epsilon_0>0$, we also prove a lower bound of $\tilde{\Omega}(n^{1/3})$
  for the query complexity,
  improving the previous best lower bound of $\tilde{\Omega}(n^{1/5})$ by Glasner and Servedio
  \cite{GlasnerServedio}.
Our upper and lower bounds are tight, up to a poly-logarithmic factor, when the distance parameter
  $\epsilon$ is a constant.
Furthermore, the same upper and lower bounds can be extended to the distribution-free testing
  of general conjunctions, and the lower bound can be extended to
  that of decision lists and linear threshold functions. 
 \end{abstract}
\newpage

\section{Introduction}

The field of property testing analyzes the resources an algorithm requires
  to determine whether an unknown object
  satisfies a certain property versus \emph{far} from satisfying the property.
It was introduced in \cite{RubinfeldSudan}, after prior work in \cite{Babai,Blum}, and
  has been studied extensively during the past two decades (see surveys in \cite{Goldreich,Fischer,Ron,AlonShapira,Rubinfeld}).

For our purpose, consider a Boolean function $f:\{0,1\}^n\rightarrow \{0,1\}$ and a class
  $\fC$ of Boolean functions, viewed as a property. 
The distance between $f$ and $\fC$ in the standard testing model is measured
  with respect to the \emph{uniform distribution}.
Equivalently, it is the smallest fraction of entries of $f$ one needs to
  flip to make it a member of $\fC$.
A natural generalization of the standard model, called \emph{distribution-free property testing},~was
  first introduced by Goldreich, Goldwasser and Ron \cite{GoldreichGoldwaserRon} and has
  been studied in \cite{AilonChazelle,HalevyKushilevitz,GlasnerServedio,HalevyKushilevitz2,HalevyKushilevitz3,DolevRon}.
In the distribution-free model, there is an unknown distribution $\calD$ over $\{0,1\}^n$
  in addition to the unknown $f$.
The goal of an algorithm is   to determine whether $f$ is in $\fC$ versus
  far from $\fC$ \emph{with respect to $\calD$}, given black-box access to $f$ and sampling access to $\calD$.
The model of distribution-free property testing is well motivated by scenarios where the distance
  being of interest
  is indeed measured with respect to an unknown distribution $\calD$. It
  is also inspired by similar models in computational learning theory (e.g.,
  the distribution-free PAC learning model \cite{Valiant} with membership queries).
It was observed \cite{GoldreichGoldwaserRon} that any proper distribution-free PAC learning algorithm
  can be used for distribution-free property testing.

In this paper we study the distribution-free testing
  of \emph{monotone conjunctions} (or \emph{monotone monomials}):
  $f$ is a monotone conjunction if $f(z)=\bigwedge_{i\in S} {z_i}$, for some
  $S\subseteq [n]$.
We first obtain an efficient algorithm that is one-sided
  and makes $\smash{\tilde{O}((n^{1/3}/\epsilon^5))}$ queries.
When $1/\epsilon$ is small compared to $n$, it improves the previous best
  $\smash{\tilde{O}(n^{1/2}/\epsilon)}$-query algorithm of Dolev and Ron \cite{DolevRon}.

\begin{theorem} \label{uppertheorem}
There is a ${O((n^{1/3}/\epsilon^5)\cdot\log^7(n/\epsilon))}$-query one-sided
  algorithm for the distribution-free\\ testing of monotone conjunctions.
\end{theorem}

For some constant distance parameter $\epsilon_0>0$,
we also present a $\tilde{\Omega}({n^{1/3}})$ lower bound on the
  number of queries required by any distribution-free testing algorithm.
This improves the previous best lower bound of $\tilde{\Omega}(n^{1/5})$ by Glasner and Servedio \cite{GlasnerServedio}.

\begin{theorem}  \label{lowertheorem}
There exists a universal constant $\epsilon_0>0$ such that any
  two-sided distribution-free algorithm\\ for testing whether an unknown Boolean
  function is a monotone conjunction versus $\epsilon_0$-far from monotone\\ conjunctions
  with respect to an unknown distribution
  must make $\Omega({n^{1/3}}/{\log^3 n})$ queries.
\end{theorem}

Notably when the distance parameter $\epsilon$ is a constant,
  our new upper and lower bounds
  given in Theorems \ref{uppertheorem} and
  \ref{lowertheorem} are tight for the distribution-free testing of
  monotone conjunctions up to a poly-logarithmic factor of $n$.
Furthermore, these bounds can also be extended to several
  other basic Boolean function classes.

First, our upper bound can be extended to general conjunctions
  (i.e. $f$ is the conjunction of a subset of literals in $\{z_1,\ldots,z_n,\overbar{z_1},\ldots,\overbar{z_n}\}$)
  via a reduction to the distribution-free testing of monotone conjucntions,
  improving the previous best
  $\smash{\tilde{O}(n^{1/2}/\epsilon)}$-query algorithm of Dolev and Ron \cite{DolevRon} when $1/\epsilon$ is small.

\begin{theorem} \label{uppertheoremforgeneral}
There is a ${O((n^{1/3}/\epsilon^5)\cdot\log^7(n/\epsilon))}$-query one-sided
  algorithm for the distribution-free\\ testing of general conjunctions.
\end{theorem}

Second, our lower bound can be extended to the distribution-free testing of
  general conjunctions, deci\-sion lists, as well as linear threshold functions
  (see their definitions in Section \ref{sec:pre}), improving the previous
  best lower bound of $\tilde{\Omega}(n^{1/5})$ by Glasner and Servedio \cite{GlasnerServedio}
  for these classes.
For general conjunctions, our bounds are also tight up to a poly-logarithmic factor of $n$ when
  $\epsilon$ is a constant.

\begin{theorem}  \label{lowertheoremgeneral}
There exists a universal constant $\epsilon_0>0$ such that any
  two-sided distribution-free algorithm\\ for testing whether an unknown Boolean
  function is a general conjunction versus $\epsilon_0$-far from
  general\\ conjunctions
  with respect to an unknown distribution
  must make $\Omega({n^{1/3}}/{\log^3 n})$ queries.
The same lower\\ bound holds for testing decision lists and testing linear threshold functions.
\end{theorem}

In most part of the paper we focus on the distribution-free testing of monotone conjunctions
  (except~for Sections
 \ref{upperboundgeneral},  \ref{lowerboundgeneral}~and  \ref{LTFsec}).~We start with some intuition behind our new algorithm and lower bound construction for monotone conjunctions, and compare our approaches and techniques
  with those of \cite{GlasnerServedio} and \cite{DolevRon}.

\subsection{The Lower Bound Approach}

We start with the lower bound because our algorithm was indeed
  inspired by obstacles we encountered when attempting to push it further
  to match the upper bound of Dolev and Ron \cite{DolevRon}.

We follow the same high-level approach of Glasner and Servedio \cite{GlasnerServedio}.
They define two distributions $\yesd$ and $\nod$:
  in each pair $(f,\calD_f)$ drawn from $\yesd$, $f$ is a monotone conjunction, whereas
  in each $(g,\calD_g)$ drawn from $\nod$, $g$ is constant-far from monotone conjunctions with respect to $\calD_g$.
Then they show that no algorithm with $\tilde{O}(n^{1/5})$ queries
  can distinguish them.
We briefly review their construction and arguments.

Both distributions start by sampling $m=n^{2/5}$ pairwise disjoint sets $C_i$ of size $n^{2/5}$ each.
Each $C_i$ is then randomly partitioned into two disjoint sets $A_i,B_i$ of the same size, with
  a special index $\alpha_i$ randomly sampled from $A_i$. Let $a^i,b^i,c^i$ denote the strings
  with $A_i=\zero(a^i)$, $B_i=\zero(b^i)$, and $C_i=\zero(c^i)$, where we write $\zero(x)=\{i:x_i=0\}$.
For $\yesd$, $f$ is the conjunction of $x_{\alpha_i}$'s, $i\in [m]$, and
   $x_j$'s, $j\notin \cup_i C_i$. So $f(a^i)=f(c^i)=0$ and $f(b^i)=1$.
$\calD_f$ puts weight $2/(3m)$ on $b^i$ and $1/(3m)$ on $c^i$.
The definition~of~$\nod$~is much more involved.
  $g$ sets $g(a^i)=g(b^i)=1$ and $g(c^i)=0$; $\calD_g$ is uniform over all $3m$ strings
  $\{a^i,b^i,c^i\}$.
On the one hand, $g$ is clearly far from monotone conjunctions with respect to $\calD_g$.
On the other hand,
by the birthday paradox, any algorithm that draws $n^{1/5}$ samples
  with high probability gets at most one sample from each triple $(a^i,b^i,c^i)$,
  and information theoretically cannot distinguish $\yesd$ and $\nod$: What the algorithm
  sees is just a bunch of pairwise disjoint sets of two sizes, as $\zero(x)$ of samples $x$ received.
In discussion below we refer to them as the sets the algorithm receives in the sampling phase
  \footnote{Without
  loss of generality, we may always assume that an algorithm starts by a sampling phase when
  it receives all the samples
  drawn from $\calD$. After that it only queries the black-box oracle.}.

The real challenge for Glasner and Servedio is to define $g$
  in $\nod$ carefully on strings of $\{0,1\}^n$ outside of $\{a^i,b^i,c^i\}$
  such that even an algorithm with access to a black-box oracle cannot distinguish them.
For~this purpose, $g$ follows $f$ by setting $g(x)=0$ whenever $x_j=0$ for some $j\notin \cup_i C_i$.
This essentially discourages a reasonable algorithm from querying
  $z$ with $z_j=0$ for some $j$ outside of the sets it received in the sampling phase: for any such $z$,
  both $f$ and $g$ return $0$ with probability $1-n^{1/5}$ so with only $n^{1/5}$ queries
  the risk~is~too high to take.
Knowing that an algorithm only  queries such strings, 
  \cite{GlasnerServedio} sets up $g$ so that an algorithm can distinguish
  $g$ from $f$ only when it incurs an event that is unlikely to happen (e.g., hitting
  $z_{\alpha_i}=0$ in some~$A_i$ with a query $z$ that has a small $\zero(z)\cap A_i$).
When events like this do not happen, the algorithm can be successfully simulated
  with no access to the black-box oracle.
This finishes the proof.

Our lower bound proof follows similar steps as those of Glasner and Servedio \cite{GlasnerServedio}.
\emph{The improvement mainly
comes from a more delicate construction of the two distributions $\yesd$ and $\nod$,
  as well as a tighter analysis on a no-black-box-query simulation of any testing algorithm with
  access to both oracles.}
The first difficulty we encountered is a dilemma in the construction: There are only $n$ indices in total
  but we want~the following three things to happen at the same time:
We need $n^{2/3}$ sets $C_i$'s so that the birthday paradox
  still applies for $n^{1/3}$ queries; We would like each $C_i$ to have size $n^{2/3}$
  to survive black-box queries; Also $\cup_i C_i$ is better small compared to $n$ so one can still
  argue that no reasonable algorithm makes any crazy black-box query $z$
  with zero entries outside of the sets it receives.
There is simply no way to satisfy all these conditions;
  Glasner and Servedio had the best parameters in place and they are tight in more than one places.

It seems that the only possible solution is to allow $C_i$'s to have significant overlap with each other.~This, however, makes the analysis more challenging, since an algorithm may
  potentially gain crucial information from
  intersections of sets it receives in the sampling phase.
Informally we first randomly pick a set $R$~of~size $n/2$ and randomly partition it into
  $n^{1/3}$ disjoint blocks of size $n^{2/3}$ each.
Each of the $n^{2/3}$ sets $C_i$'s consists of $2\log^2 n$ random blocks and
  two special indices $\alpha_i$ and $\beta_i$ that are unique to $C_i$.
Each $C_i$ is then partitioned into~$A_i$, $B_i$ with $\log^2 n$ blocks each,
  which also receive $\alpha_i$ and $\beta_i$, respectively.
An important property from our~setup (and simple calculation) that is crucial
  to our analysis later on is that even with
  $\tilde{O}(n^{1/3})$ sets drawn uniformly,
  most likely only a $o(1)$-fraction of each set is covered by other sets sampled.
The rest of our $\yesd$ and $\nod$ is similar to \cite{GlasnerServedio},
  but with a more intricate $\nod$ function $g$ outside of the support of $\calD_g$.

Our distributions $\yesd$ and $\nod$ work well against any algorithm with no
  access to the black-box oracle.
The technically most challenging part is to show that
  any given algorithm can be simulated closely without the black-box oracle.
Note that $\cup_i C_i$ above is about $n/2$.
An algorithm with $n^{1/3}$ many queries has a much stronger incentive to take the risk
  and query $z$ with $z_j=0$, for some $j$ outside of the sets sampled.~This then demands a more sophisticated analysis  to
  characterize every possible loophole an algorithm may explore,~in distinguishing the
  two distributions $\yesd$ and $\nod$.
At the end, we need to fine-tune the construction of $\nod$
  to really fit the analysis perfectly (not surprising given the upper bound) so that we
  can manage to bound the probability of each loophole, and
  show that the no-black-box-query simulation succeeds most of the time.

\subsection{The Approach of Our Algorithm}

We now describe the high-level approach of our algorithm.
For clarity, we assume here that $\epsilon$ is a constant.
We first review the $\tilde{O}(n^{1/2})$-query algorithm of Dolev and Ron \cite{DolevRon}.
An ingredient from \cite{DolevRon}, which we also use heavily as a subroutine,
  is a deterministic binary search procedure:
upon $\smash{x\in f^{-1}(0)}$, it attempts to find an index $i\in \zero(x)$
  such that $f(\{i\})=0$.\footnote{For convenience we extend $f$ to subsets of $[n]$,
    with $f(A)$ defined as $f(z)$ with $A=\zero(z)$.}
If it fails on $x$, then $f$ is not a monotone conjunction; otherwise, let
  $h(x)$ denote the index found, called the representative index of $x$ \cite{DolevRon}.
Roughly speaking, the algorithm of Dolev and Ron draws $n^{1/2}$ samples from $\calD$
  and uses the binary search procedure to compute the representative index
  $h(x)$ of each sample $x$ from $f^{-1}(0)$.
Then the algorithm rejects if $y_\alpha=0$ for some sample $y\in f^{-1}(1)$
  and some representative index $\alpha$ found.
The algorithm is one-sided. But to reject with high probability when $f$ is far from monotone conjunctions
  with respect to $\calD$, $n^{1/2}$
  samples seem necessary.

Our algorithm was inspired by obstacles encountered when trying to
  improve the $\tilde{\Omega}(n^{1/3})$ lower bound.
To give some intuition, consider the same distribution of triples of sets, $(A_i,B_i,C_i)$,
  drawn as in the lower bound proof sketch, with $m\approx n^{2/3}$ many $C_i$'s each of size $\approx n^{2/3}$.
Let $\calD$ be the uniform distribution over $\{a^i,b^i,c^i\}$,
  with $g$ satisfying $g(a^i)=g(b^i)=1$ and $g(c^i)=0$.
Now consider the following scenario~where an adversary tries to fill in entries of $g$ outside
  of $\{a^i,b^i,c^i\}$,  aiming
  to fool algorithms with a small number of queries as a monotone conjunction.
An obstacle for the adversary is the following testers: Let $t\approx  n^{1/3}$.
\begin{flushleft}\begin{enumerate}
\item[] \textbf{Tester 1.} Draw $t$ samples $y^1,\ldots,y^t$ from $g^{-1}(1)$ with respect to $\calD$.
Let $E_i=\zero(y^i)$, $E=\cup_i E_i$.
Given the definition of $\calD$ and that $g(a^i)=g(b^i)=1$ and $g(c^i)=0$, each $E_i$ is either $A_k$ or $B_k$.
Repeat $t$ times: pick a subset $Z$ of $E$ of size $t$ uniformly at random and query $z$ with $\zero(z)=Z$.
(Note that if $g$ is a monotone conjunction, then $E$ cannot contain any index of a variable that belongs to the conjunction and hence for every $Z\subseteq E$ and $z$ with $\zero(z)=Z$, $g$ must return $1$ to query $z$.)\vspace{-0.04cm}

\item[] \textbf{Tester 2.} Draw $t-1$ samples $y^1,\ldots,y^{t-1}$ from $g^{-1}(1)$ with respect to $\calD$,
  and one sample $x$ from $g^{-1}(0)$ (so $\zero(x)=C_k$ for some $k$). Define $E_i$ and $E$ similarly.
Use the binary search procedure to find the representative index $h(x)$ of $x$;
  for the sake of discussion here assume that it finds the special index $\alpha_k$ in $C_k$
  if $\zero(x)=C_k$ (reject if $\alpha_k\in E$).
Pick a subset $Z$ of $E$ of size $t-1$ uniformly at random, and query $z$ with $\zero(z)=Z\cup \{\alpha_k\}$.
(Note that if $g$ is a monotone conjunction, then $h(x)$ must be the index of a variable
  in the conjunction  and hence, we have $h(x)\notin E$ and for every $Z\subseteq E$ and $z$ with $\zero(z)=Z\cup \{h(x)\}$, $g$ must return 0 to query $z$.)
\end{enumerate}\end{flushleft}

Consider an algorithm that runs both testers with independent samples.
Clearly $g$ fails and gets rejected if it returns $0$ to a query $z$ from Tester 1 or it returns $1$ to a
  query $z$ from Tester 2.
It turns out that there is no way to design a $g$ that returns the correct bit
  most of the time for both testers.
To see this is the case, assume for now that about half of the $E_i$'s in Tester 1
  are indeed $A_i$'s so each of them contains a special and unique index
  $\alpha_i$; in total there are $\Omega(t)$ many of them in $E$.
Given that $|E|\le n$, and we repeat $t$ times in picking $z$, most likely one of the
  strings $z$ queried has an $\alpha_i\in \zero(z)$ and it is also the only index in
  $\zero(z)\cap E_i^*$, where we let $E_i^*$ denote the indices that are unique to $E_i$ among all $E_j$'s.
(For the latter, the intuition is that there simply cannot be too many
  large $E_i^*$ because they are disjoint and their union is $E$.)

For such a string $z$ drawn and queried in Tester 1,
  $g$ has to return $1$. 
However, the distribution of such $z$ is very similar to the distribution of
  $z$ queried in Tester 2, where an $\alpha_k$ is first picked randomly (by drawing
  a $C_k$ and running the binary search procedure on it to reveal $\alpha_k$)
  and then unioned with a set of $t-1$ indices drawn uniformly from $E$ obtained from
  $t-1$ samples from $g^{-1}(1)$.

This is essentially how our algorithm works. It consists of two stages, each of which implements~one~of the two testers.
The main challenge for us is the analysis to show that it works for any input pair $(f,\calD)$
  that not necessarily looks like those constructed in $\nod$.
At a high level, we show that if $f$ is far from monotone conjunctions  with respect to $\calD$
  and passes stage 1 with high probability,
  then it fails stage 2~and~gets rejected with high probability
  since the two distributions of $z$ queried in the two stages are very close to each other.

An important ingredient of our analysis is the notion of a \emph{violation bipartite graph} $G_f$
  of a pair $(f,\calD)$.
Compared to the \emph{violation hypergraph} $H_f$ introduced by Dolev and Ron, our
  bipartite graph $G_f$ is~easier~to work with and its vertex covers
  also characterize the distance between $f$ and monotone conjunctions (similar to
  the violation hypergraph of \cite{DolevRon}).
In particular, our analysis of correctness
  heavily relies on a \emph{highly regular} bipartite subgraph $G_f^*$ of $G_f$,
  of which every vertex cover still has total weight $\Omega(\epsilon)$.
The regularity of $G_f^*$ plays a critical role in our comparison of the two stages.
More specifically, it helps bound the double counting when we lower bound the
  probability of $(f,\calD)$ failing stage 2, assuming that it passes stage 1 with high probability.
\vspace{0.26cm}

\noindent\textbf{Organization.} We define the model of
  distribution-free testing, and introduce some useful notation in~Section \ref{sec:pre}.
We present the new algorithm for monotone conjunctions and its analysis in Section \ref{sec:upper},
  followed~by the lower bound proof in Section \ref{sec:lower}.
We then extend the upper bound to  general conjunctions in Section
  \ref{upperboundgeneral}, and extend the lower bound to general conjunctions
  and decision lists in Section \ref{lowerboundgeneral}, and to linear threshold functions in
  Section \ref{LTFsec}.

\section{Preliminaries}\label{sec:pre}

We review the model of
  distribution-free property testing and then introduce some useful notation.

Let $f,g:\{0,1\}^n\rightarrow \{0,1\}$ denote two Boolean functions over
  $n$ variables, and $\calD$ denote a probability distribution over $\{0,1\}^n$.
We define the distance between $f$ and $g$ with respect to $\calD$ as
$$
\dist_\calD(f,g)=\Pr_{{z}\in \calD} \big[f({z})\ne g(z)\big].
$$
Given a class $\fC$ of Boolean functions over $\{0,1\}^n$, we define
$$
\dist_\calD(f,\fC)=\min_{g\in \fC}\Big(\dist_\calD(f,g)\Big)
$$
as the distance between $f$ and $\fC$ with respect to $\calD$.
We also say $f$ is $\epsilon$-far from $\fC$ with respect to $\calD$ for some $\epsilon\ge 0$
  if $\dist_\calD(f,\fC)\ge \epsilon$.
Now we define distribution-free testing algorithms. 

\begin{definition}
Let $\fC$ be a class of Boolean functions over $\{0,1\}^n$.
A distribution-free testing algorithm~$T$ for $\fC$ is a probabilistic
  oracle machine with access to a pair $(f,\calD)$, where
  $f$ is an unknown Boolean function $f:\{0,1\}^n\rightarrow \{0,1\}$
  and $\calD$ is an unknown probability distribution   over $\{0,1\}^n$, via
\begin{enumerate}
\item a black-box oracle that returns the value $f(z)$ when $z\in \{0,1\}^n$ is queried; and \vspace{-0.2cm}
\item a sampling oracle that returns a pair $(z,f(z))$
with $z$ drawn independently from $\calD$ each time.
\end{enumerate}
The algorithm $T$ takes as input a distance parameter $\epsilon>0$
  and satisfies
for any $(f,\calD)$:
\begin{enumerate}
\item If $f\in \fC$, then $T$ accepts with probability at least $2/3$; and\vspace{-0.2cm}
\item If $f$ is $\epsilon$-far from $\fC$ with respect to $\calD$,
  then $T$ rejects with probability at least $2/3$.
\end{enumerate}
We say an algorithm is \emph{one-sided} if it always accepts a function $f$ in $\fC$.
\end{definition}

In this paper we focus on the distribution-free testing of
  $\mconj$, the class of all monotone conjunctions (or monotone monomials as in
  \cite{DolevRon}): $f:\{0,1\}^n\rightarrow \{0,1\}$ is in $\mconj$
  if there exists an $S\subseteq [n]$ with\vspace{-0.06cm}
$$
f(z_1,\ldots,z_n)=\bigwedge_{i\in S} z_i.\vspace{-0.06cm}
$$
Note that $f$ is the all-$1$ function when $S$ is empty.
{In addition to monotone conjunctions we are
  interested~in the distribution-free testing of general conjunctions, decision lists, and linear threshhold functions:
\begin{flushleft}\begin{itemize}
\item We say $f:\{0,1\}^n\rightarrow \{0,1\}$ is a general conjunction if there exist two sets
  $S,S'\subseteq[n]$ with\vspace{-0.08cm}
$$
f(z_1,\ldots,z_n)=\left(\bigwedge_{i\in S} z_i\right)\bigwedge \left(\bigwedge_{i\in S'} \overbar{z_i}\right).\vspace{-0.5cm}
$$.
\item A decision list $f:\{0,1\}^n\rightarrow \{0,1\}$ of length $k$ over
  Boolean variables $z_1,\ldots,z_n$ is
  defined by a sequence of $k$ pairs $(\ell_1,\beta_1),\ldots,(\ell_k,\beta_k)$ and a bit
  $\beta_{k+1}$, where $\beta_i\in\{0,1\}$ for all  $i\in [k+1]$ and
  each $\ell_i$ is a literal in $\{z_1,\ldots,z_n,\overbar{z_1},\ldots,\overbar{z_n}\}$.
Given any $z\in \{0,1\}^n$,
   $f(z)$ is determined in the following way:
  $f(z)=\beta_i$ if $i\in [k]$ is the smallest index such that $\ell_i$ is made true by $z$;
  if no $\ell_i$ is true then $f(z)=\beta_{k+1}$.

\item We say $f:\{0,1\}^n\rightarrow \{0,1\}$ is a linear threshold function
 if there exist $w_1,w_2,\ldots,w_n,\theta\in \mathbb{R}$ such that $f(z)=1$ if
$w_1z_1+\cdots+w_nz_n\ge\theta$ and $f(z)=0$ if $w_1z_1+\cdots+w_nz_n<\theta$.
\end{itemize}\end{flushleft}
}

Next we introduce some notation used in the proofs.
Given a positive integer $n$ we let $[n]=\{1,\ldots,n\}$.
Given a distribution $\calD$ over $\{0,1\}^n$   we use $\calD(z)$ to
  denote the probability of a string $z$ in $\{0,1\}^n$
  and $\calD(C)$ to denote the total probability of strings in $C\subseteq \{0,1\}^n$.

We call $x$ a $0$-string (with respect to $f$) if $f(x)=0$, and write $f^{-1}(0)$ to
  denote the set of $0$-strings.
We call $y$ a $1$-string (with respect to $f$) if $f(y)=1$, and write $f^{-1}(1)$ to
  denote the set of $1$-strings.

For both our lower and upper bound proofs, it is easier to
  use the language of sets.
Given $z\in \{0,1\}^n$:\vspace{-0.1cm}
$$
\zero(z)=\big\{i\in [n]: z_i=0\big\}.\vspace{-0.1cm}
$$
For convenience we abuse the notation and allow   $f$ to take as input
  a subset of $[n]$: $f(E)$ is defined as $f(z)$ with $z\in \{0,1\}^n$ and
  $E=\zero(z)$.
This should be clear from~the context, since we use lowercase letters
  for strings and uppercase letters for sets.
We call $A$ a $0$-set if~$f(A)=0$, and $B$ a $1$-set if $f(B)=1$.

We use $\11^n$ to denote the all-$1$ string of length $n$ and drop the $n$ when
  it is clear from the context.

\section{Upper Bound: Proof of Theorem \ref{uppertheorem}}\label{sec:upper}

In this section, we present our one-sided distribution-free tester
  for $\mconj$.
Throughout the section we use $f:\{0,1\}^n\rightarrow \{0,1\}$ to denote
  the unknown Boolean function, and $\calD$ to denote the unknown distribution.

For clarity of the analysis in this section, we always write $x$ to denote a string from $f^{-1}(0)$,
  $y$ to denote a string from $f^{-1}(1)$, and $z$ to denote a string with $f(z)$ unknown
  (or we do not care about $f(z)$).

\subsection{Binary Search, Empty Strings, and Representative Indices}

The algorithm of Dolev and Ron \cite{DolevRon} uses a deterministic binary search procedure which,
  given a string $x\in f^{-1}(0)$, tries to find an index $i\in \zero(x)$
  such that $f(\{i\})=0$.
(Note that such an $i$ always exists if $f$ is in $\mconj$.)
Our algorithm also uses it as a subroutine  so
  we include it in Figure \ref{BinarySearch} for completeness.

We record the following property of the binary search procedure:

\begin{property}\label{simpleproperty}
The binary search procedure uses $O(\log n)$ many queries.
Given as an input $x\in f^{-1}(0)$, it returns either $\textsf{\emph{nil}}$ or an index $i\in \zero(x)$ such
  that $f(\{i\})=0$.
The former never happens if $f\in \mconj$.
\end{property}

Given $x\in f^{-1}(0)$, we write $h(x)\in [n]\cup \{\nil\}$ to denote the
   output of the binary search procedure~on $x$ ($h(\cdot)$ is well-defined since the procedure is deterministic).
We follow \cite{DolevRon} and call $x\in f^{-1}(0)$ an \emph{empty} string~(with respect to $f$) if $h(x)=\nil$,
  and call $h(x)\in [n]$ the \emph{representative index} of $x$ (with respect to $f$) when $h(x)\ne \nil$.

\begin{figure}[t!]
\textbf{Algorithm 1: Binary Search.} Input: $x\in f^{-1}(0)$.
\begin{flushleft}\begin{enumerate}
\item Let $Z=\zero(x)$. If $Z=\emptyset$, return \textsf{nil}; if $|Z|=1$, output the only
  index in $Z$.\vspace{-0.1cm}

\item While $|Z|\ge 2$ do\vspace{-0.1cm}
\begin{itemize}
\item[--] Let $Z_0$ be the subset of $Z$ that contains the smallest $\lceil |Z|/2\rceil$ indices in $Z$,
  and  $Z_1=Z\setminus Z_0$.\vspace{-0.06cm}
\item[--] Query both $f(Z_0)$ and $f(Z_1)$.\vspace{-0.06cm}
\item[--] If $f(Z_0)=0$, set $Z=Z_0$; if $f(Z_0)=1$ but $f(Z_1)=0$, set $Z=Z_1$; otherwise, return \textsf{nil}.\vspace{-0.08cm}
\end{itemize}
\item Return the only element that remains in $Z$.\vspace{-0.2cm}
\end{enumerate}\end{flushleft}
\caption{The binary search procedure from Dolev and Ron \cite{DolevRon}.\vspace{0.08cm}}\label{BinarySearch}
\end{figure}

\subsection{A One-sided Algorithm for Testing Monotone Conjunctions}

We use the following parameters in the algorithm and its analysis:
\begin{equation}\label{parameters1}
{d=\frac{\log^2 (n/\epsilon)}{\epsilon},\quad\ d^*=d^2/\epsilon,\quad\  r=n^{1/3},\quad\
 t= d\cdot r \quad\ \text{and}\quad \ s=t\log n.}
\end{equation}

Our algorithm is presented in Figure \ref{fig:main-algorithm}, which consists
  of three stages.
We refer to it as \emph{Algorithm 2} and
  start its analysis with the following simple observations.

\begin{observation}\label{obob1}
The number of queries used by Algorithm 2 is ${O((n^{1/3}/\epsilon^5)\cdot\log^7(n/\epsilon)).}$
\end{observation}
\begin{observation}\label{obob2}
All queries to the sampling oracle are made in Stage 0.
\end{observation}

Next we prove that this is indeed a one-sided algorithm for testing
  monotone conjunctions.

\begin{lemma}\label{lem:onesided}
If $f\in \mconj$, then Algorithm 2 always accepts $(f,\calD)$ for any distribution
  $\calD$ over $\{0,1\}^n$.
\end{lemma}
\begin{proof}
Since Algorithm 2 always accepts at the end of Stage 2,
  it suffices to show that it never rejects when $f$ is a
  monotone conjunction.
First note that $\smash{f(\11^n)}$ must be $1$ when $f$ is a monotone conjunction.
By Property \ref{simpleproperty}, $h(x)=\nil$ can never happen in Stage 0
  when $f$ is a monotone conjunction and $x\in f^{-1}(0)$.

This leaves us to check lines 1.1, 1.2, 2.1 and 2.2.
Assume that $f\in \mconj$:
\begin{flushleft}\begin{enumerate}
\item If $B_1,...,B_k\subseteq [n]$ satisfy $f(B_1)=\cdots =f(B_k)=1$,
  then every $Z\subseteq \cup_i B_i$ satisfies $f(Z)=1$.
\\ This implies that Algorithm 2 never rejects on line 1.1, 1.2 or 2.1.\vspace{-0.1cm}

\item For line 2.2, $\alpha=h(x)$ implies that $f(\{\alpha\})=0$ which
  implies that $f(P\cup\{\alpha\})=0$ when $f$ is\\ a monotone conjunction.
So Algorithm 2 never rejects on line 2.2.
\end{enumerate}\end{flushleft}
This finishes the proof of the lemma.
\end{proof}

\begin{figure}[t!]
\textbf{\ \ \ \ Algorithm 2: Monotone Conjunctions.\vspace{-0.06cm}}
\begin{spacing}{1.04}\begin{flushleft}\begin{itemize}
\item[] \textbf{Stage 0.}
Query $f(\11^n)$ and \textsf{Reject} if $f(\11^n)=0$.
Make $3t(d^*+1)/\epsilon$ many queries to the sampling\\ oracle.
Let $(z^{i,j},f(z^{i,j}))$ denote the pairs received, for $i\in [d^*+1]$ and $j\in [3t/\epsilon]$.
Run the binary search procedure
  to compute the representative index $h(x)$ for each $x\in f^{-1}(0)$ sampled.
\textsf{Reject}\\ if one of them has $h(x)=\nil$.

\item[] \textbf{Stage 1.}
\textsf{Accept} if the number of $j\in [3t/\epsilon]$ with $z^{1,j}\in f^{-1}(1)$ is less than $t$; 
  otherwise, we let $y^1,\ldots,y^t$ be
  the first $t$ (not necessarily distinct) $1$-strings in $(z^{1,j})$.
 Let $B_i=\zero(y^i)$, $B=\cup_i B_i$.\vspace{0.1cm}

\subitem \hspace{-0.3cm}1.1. Repeat $s$ times: Draw an index $i$ from $B$ uniformly at random.
  \textsf{Reject} if $f(\{i\})=0$.\vspace{-0.04cm}

\subitem \hspace{-0.3cm}1.2. Repeat $s$ times: Draw
  a subset $Z\subseteq B$ of size $r$ uniformly at random. \textsf{Reject} if $f(Z)=0$.
 \vspace{0.08cm}

\item[] \textbf{Stage 2.} Repeat the following steps for $d^*$ iterations.
For the $i$th iteration, $i\in [d^*]$:\vspace{0.015cm}
\newline \textsf{Accept} if the number of $j\in [3t/\epsilon]$ with $z^{i+1,j}\in f^{-1}(1)$
  is less than $t-1$ or  no string in $(z^{i+1,j})$ is from $f^{-1}(0)$;
otherwise, let $y^1,\ldots,y^{t-1}$ be the first $t-1$
  (not necessarily distinct) $1$-strings from $(z^{i+1,j})$,
  and $x$ be the first $0$-string from $(z^{i+1,j})$.
Let $B_i=\zero(y^i)$ for each $i$, and $B=\cup_i B_i$.
\\ Use the binary search procedure to compute $h(x)$, and
\textsf{Reject} if $h(x)=\nil$.\vspace{0.1cm}

\subitem\hspace{-0.3cm}2.1\hspace{0.1cm} Let $\alpha=h(x)\in \zero(x)$.
\textsf{Reject} if $\alpha\in B$.

\subitem\hspace{-0.3cm}2.2. 
Uniformly draw a $P\subseteq B$
  of size $r-1$. \textsf{Reject} if
  $\smash{f(P\cup\{\alpha\})=1}$.\vspace{0.06cm}

\item[] \textbf{End of Stage 2.} \textsf{Accept}.\vspace{-0.3cm}
\end{itemize}\end{flushleft}\end{spacing}
\caption{The distribution-free
algorithm for testing monotone conjunctions.}\label{fig:main-algorithm}
\end{figure}

Theorem \ref{uppertheorem} follows directly from the following lemma combined with
  Observation \ref{obob1} and Lemma \ref{lem:onesided} (since Algorithm 2 is
  one-sided its success probability in Lemma \ref{mainlemma} can be easily amplified to $2/3$).

\begin{lemma}\label{mainlemma}
If $f$ is $\epsilon$-far from $\mconj$ with respect to $\calD$,
  Algorithm 2 rejects with probability at least $0.1$.
\end{lemma}

\subsection{Reduction to Well-Supported Probability Distributions}

To ease the proof of Lemma \ref{mainlemma}, we show that it suffices
  to focus on so-called well-supported distributions.
We say a probability distribution $\calD$ on $\{0,1\}^n$ is \emph{well-supported} with respect
  to $f$ if every empty string~of~$f$ has probability zero in $\calD$.
Given $f$, intuitively an adversary to pair it with a hard probability
distribution $\calD$ may not
  want to allocate much probability on empty points of $f$, in case Algorithm 2
  rejects in Stage 0.

Following this intuition that well-supported probability distributions are probably
  hard cases of Lemma \ref{mainlemma},
  we prove Lemma \ref{main-lemma-2} below concerning
  such distributions in the rest of the section.
Before its proof we show that it indeed implies Lemma \ref{mainlemma}.

\begin{lemma}\label{main-lemma-2}
Assume that $f$ is a Boolean function and $\calD'$ is a well-supported
  distribution with respect to $f$.
If $f$ is $(\epsilon/2)$-far from $\mconj$ with respect to $\calD'$,
  Algorithm 2 rejects $(f,\calD')$ with probability at least $0.1$.
\end{lemma}

\begin{proof}[Proof of Lemma \ref{mainlemma} assuming Lemma \ref{main-lemma-2}]
Assume that $f$ is $\epsilon$-far from $\mconj$ with respect to $\calD$.
Let $\delta\ge 0$ denote the total probability of $\calD$ over
  empty strings of $f$.
If $\delta=0$, Lemma \ref{mainlemma} follows directly from Lemma \ref{main-lemma-2}
  since $\calD$ is well-supported.
If $\delta\ge \epsilon/2$, Algorithm 2 must reject with probability
  $1-o(1)$ in Stage 0.
We consider below the remaining case when $0<\delta<\epsilon/2$.

Let $\calD'$ denote the following distribution derived from $\calD$.
The probability of any empty string of $f$ in $\calD'$ is $0$.
The probability of any other string is set to be its probability in $\calD$
  multiplied by $1/(1-\delta)$.
Clearly $\calD'$ is now a well-supported probability distribution
  with respect to $f$.
We prove the following claim:

\begin{claim}\label{claim-claim}
The probability of Algorithm 2 rejecting $(f,\calD)$ is at least as large
  as that of rejecting $(f,\calD')$.
\end{claim}
\begin{proof}
Algorithm 2 always rejects $(f,\calD)$ if one of
  the samples in Stage 0 is an empty string.
Let $E$ denote the event that no sample in Stage 0 is empty.
Then the probability of Algorithm 2 accepting $(f,\calD')$
  is exactly that of it accepting $(f,\calD)$ conditioning on $E$.
This follows from the definition of $\calD'$ and our observation~\ref{obob2}:
  Stages 1 and 2 access the black-box oracle only, which does not involve
  $\calD$ or $\calD'$.
As a result, we have
\begin{align*}
\Pr\big[ \text{$(f,\calD)$ accepted}\hspace{0.04cm}\big]
 =\Pr\big[\text{$(f,\calD)$ accepted}\hspace{0.04cm}\big|\hspace{0.04cm}E\hspace{0.04cm}\big]\cdot \Pr[E]
 \le \Pr\big[\text{$(f,\calD')$ accepted}\hspace{0.04cm}\big].
\end{align*}
This finishes the proof of the claim.
\end{proof}

Finally we show that $f$ is $(\epsilon/2)$-far from $\mconj$
  with respect to $\calD'$.
Given this we can then apply Claim \ref{claim-claim} to finish the proof of the lemma.
To see this is the case, note that the total variation distance
  $d_{TV}(\calD,\calD')$ is $\delta$ by the definition of $\calD'$.
This implies that for any Boolean function $g$, we have
$$\big|\hspace{0.03cm}\dist_\calD(f,g)-\dist_{\calD'}(f,g)\hspace{0.03cm}\big|\le d_{TV}(\calD,\calD')\le \delta.$$
As a result, $\dist_{\calD'}(f,\mconj)\ge \dist_{\calD}(f,\mconj)-\delta\ge \epsilon/2.$
This finishes the proof of the lemma.
\end{proof}

We prove Lemma \ref{main-lemma-2} in the rest of the section.
For convenience, we still use $\calD$ to denote the unknown distribution,
  but from now on we always assume without loss of generality that 1) $\calD$ is well-supported
  with respect to $f$, and 2) $f$ is $(\epsilon/2)$-far from $\mconj$
  with respect to $\calD$.

It is worth mentioning that since $\calD$ is well-supported, Algorithm 2 can skip Stage 0,
  which is the reason why it is named Stage 0,
  and have both Stage 1 and each iteration of Stage 2
  start by making $3t$ new queries to the sampling oracle.
We will follow this view in the analysis of Algorithm 2 in the rest of the section.

\def\edgewt{\wt}

\subsection{The Violation Bipartite Graph}

We first review the \emph{violation hypergraph} of a
  Boolean function $f$ introduced by Dolev and Ron \cite{DolevRon}. It inspires us to define the
 \emph{violation bipartite graph} $G_f$ of $f$. The latter
  is conceptually simpler, and characterizes the distance of $f$
  to $\mconj$ as well.
The main lemma of this subsection shows that if $\dist_\calD(f)\ge \epsilon/2$,
  then $G_f$ has a \emph{highly regular} subgraph $G_f^*$ with vertex covers of weight $\Omega(\epsilon)$ only.

We start with the definition of the violation hypergraph of
  a given $f:\{0,1\}^n\rightarrow \{0,1\}$ from \cite{DolevRon}.

\begin{definition}[Violation Hypergraph]
Given $f$, we call $H_f=(V(H_f),E(H_f))$ the \emph{violation hypergraph}\\
  of $f$, where
  $V(H_f)=\{0,1\}^n$;
  $E(H_f)$ contains all subsets $\{x,y^1,\ldots,y^t\}\subseteq \{0,1\}^n$
  such that
\begin{itemize}
\item[--] $f(x)=0$; $f(y^i)=1$ for all $i:1\le i\le t$; and
  $\zero(x)\subseteq \cup_{i=1}^t \zero(y^i)$.
\end{itemize}
Note that $\{\11^n\}\in E(H_f)$ if $f(\11^n)=0$ \emph{(}this is the only
  possible special case when $t=0$\emph{)}.
\end{definition}

It turns out that $\dist_\calD(f,\mconj)$ is characterized by
   weights of vertex covers of $H_f$.

\begin{lemma}[Lemmas 3.2 and 3.4 of \cite{DolevRon}]\label{LemTotal}
A function $f$ is in $\mconj$ if and only if $E(H_f)=\emptyset$.\newline
For any Boolean function $f$,
  every vertex cover $C$ of $H_f$ has total probability $\calD(C)\ge \dist_\calD(f,\mconj)$.
\end{lemma}

Note that this lemma holds for any (not necessarily well-supported) probability distribution $\calD$.
Now we define the violation bipartite graph of $f$.

\begin{definition}[Violation Bipartite Graph]
Given a Boolean function $f$
  we call the following graph $G_f=(L\cup R,E)$ the \emph{violation bipartite graph}
  of $f$: vertices on the left side are $L=f^{-1}(1)$; vertices on the right side are
  $R=\{j\in [n]:x\in f^{-1}(0)\ \text{and}\ h(x)=j\}$;
  add an edge between $y\in f^{-1}(1)$ and $j\in R$ if $y_j=0$.

Let $\calD$ be a probability distribution over $\{0,1\}^n$.
It defines a nonnegative weight $\mathrm{wt}_\calD(\cdot)$ for each vertex in $G_f$ as follows.
The weight of $y\in f^{-1}(1)=L$ is simply $\mathrm{wt}_\calD(y)=\calD(y)$.
The weight of $j\in R$  is
$$\mathrm{wt}_\calD(j)=\sum_{x\in f^{-1}(0):\hspace{0.05cm}h(x)=j} \calD(x).$$
Given a set of vertices $C\subseteq L\cup R$, we let
  $\mathrm{wt}_\calD({C})$ denote the total weight of $C$: $\mathrm{wt}_\calD({C})=\sum_{u\in C}
  \mathrm{wt}_\calD(u)$.
Most of the time when $\calD$ is clear from the context, we drop the subscript
  and use simply $\mathrm{wt}$ for the weight. 
\end{definition}

From now on we assume that $\calD$ is well-supported with respect to $f$.
We get the following corollary: 

\begin{corollary}\label{totalweight}
If $\calD$ is well-supported,
  then every vertex cover $C$ of $G_f$ has $\wt{C}\ge \dist_\calD(f,\mconj)$.
\end{corollary}
\begin{proof}
Given a vertex cover $C$ of $G_f$, we define a vertex cover $C'$ of $H_f$
  as follows.
$C'$ consists of 1) all the empty strings of $f$; 2) $C\cap L=C\cap f^{-1}(1)$;
  and 3) $x\in f^{-1}(0)$ such that $h(x)\ne \nil$ and $h(x)\in C\cap R$.

By the definition of $C'$ and $\wt{\cdot}$, we have
  $\wt{C}=\calD(C')$ ($\calD$ is well-supported so has zero probability
  on empty strings).
It suffices to show that $C'$ is a vertex cover of $H_f$, and then apply Lemma \ref{LemTotal}.

Fix a hyperedge $\{x,y^1,\ldots,y^t\}$ in $H_f$.
For the special case when $t=0$, we have $x=\11^n$ and $f(\11^n)=0$.
So $\11^n$ is empty, and $\11^n\in C'$.
When $t\ge 1$, either $h(x)=\nil$, for which case we have $x\in C'$,
  or $h(x)\ne \nil$ and $h(x)\in \zero(x)$.
The latter implies $h(x)\in \zero(y^k)$, for some $k\in [t]$,
  and thus, $(y^k,h(x))$ is an~edge in $G_f$.
Since $C$ covers this edge, either $y^k\in C'$
  or $x\in C'$.
This finishes the proof of the lemma.
\end{proof}

Next, we extract from $G_f$ a highly regular bipartite graph $G_f^*$, with the
  guarantee that any vertex cover of $G_f^*$ still has total weight $\Omega(\epsilon)$
  (recall that $\dist_\calD(f,\mconj)\ge \epsilon/2$).
We start with some notation.
Given a subgraph $G=(L(G)\cup R(G),E(G))$ induced by
  $L(G)\subseteq L$ and $R(G)\subseteq R$, the \emph{weight} of
  graph $G$ is
$$
\edgewt{G}=\sum_{y\in L(G)} \wt{y}\cdot \deg_G(y).
$$
where $\deg_G(y)$ is the degree of $y$ in $G$.
Equivalently, one can assign each edge $(y,j)$ in $G_f$
  an edge weight of $\wt{y}$, and $\edgewt{G}$ is its total edge weight.
For each $j\in R(G)$, we define its \emph{incoming weight} as
$$
\inwt(j)=\sum_{y:\hspace{0.04cm}(y,j)\in E(G)} \wt{y},
$$
which can be viewed as the total edge weight from edges incident to $j$.

Recall the parameter
$d$ in (\ref{parameters1}).
We say a vertex $y\in L(G)$ is \emph{heavy} in $G$ if
$
\deg_G(y)\ge d\cdot \edgewt{G}; $
a vertex $j\in R(G)$ is \emph{heavy} in $G$ if
$
 \inwt(j)\hspace{-0.02cm}\ge\hspace{-0.02cm} d\cdot \edgewt{G}\cdot \wt{j}.
$
In either cases, removing a heavy vertex $u$ (and
  its incident edges) would reduce $\edgewt{G}$ by
  $\ge d\cdot \edgewt{G}\cdot \wt{u}$.
We say a vertex is \emph{light} if it is not heavy.

We run the following deterministic procedure on $G_f$
  to define a subgraph $G_f^*$ of $G_f$.
(This procedure is not new and has seen many applications in the literature,
  e.g., see \cite{RazMcKenzie}.)
\begin{flushleft}\begin{enumerate}
\item Let $G=G_f$ and $S=\emptyset$. Remove all vertices in $G$ with degree zero.\vspace{-0.1cm}
\item Remove all heavy vertices on the left side of $G$
  and their incident edges, if any; \\move them to $S$.
Also remove vertices on the right side that now have degree zero.\vspace{-0.1cm}
\item If $G$ has a vertex cover $C$ of total (vertex) weight $\wt{C}\le \epsilon/4$, exit.\vspace{-0.1cm}
\item Remove all heavy vertices on the right side of $G$ and their incident edges,
if any;\\ move them to $S$.
  Also remove vertices on the left side that now have degree zero.\vspace{-0.1cm}
\item If $G$ has a vertex cover $C$ of total (vertex) weight
  $\wt{C}\le \epsilon/4$ or there exists no more\\ heavy vertex in $G$, exit;
  otherwise go back to Step 2.
\end{enumerate}\end{flushleft}
Let $G_f^*=(L^*\cup R^*,E^*)$ denote the subgraph of
  $G_f$ induced by $L^*\subseteq L$ and $R^*\subseteq R$ we obtain at the end.

We show that $G_f^*$ has no heavy vertex, and any vertex cover $C$ of $G_f^*$
  still has a large total weight.

\begin{lemma}\label{vcvclem}
Assume that $\calD$ is well-supported with respect to $f$
  and they satisfy $\dist_\calD(f,\mconj)\ge \epsilon/2$.\\
Then $G_f^*$ has no heavy vertex, and any
  of its vertex cover $C$ has a total weight of $\wt{C}\ge 3\epsilon/8$.
\end{lemma}
\begin{proof}
The first part, i.e. $G_f^*$ has no heavy vertex,
  follows from the second part of the lemma, which implies that
  the procedure exits because $G$ contains no more heavy vertex.

The second part follows from the claim that $\wt{S}=o(\epsilon)$ (as for
  any vertex cover $C$ of~$G_f^*$, $C\cup S$ is
  a vertex cover of $G_f$ but by Corollary \ref{totalweight},
  $\wt{C\cup S}\ge \epsilon/2$).
To prove the claim,  we let $G_0, \ldots,G_s$ denote the sequence
  of graphs obtained by following the procedure, with $G_f=G_0$ and $G_s=G_f^*$,
  and let $S_i$ denote the set of vertices that are removed from $G_i$ to obtain $G_{i+1}$
  and added to $S$.
(Note that $S_i$ does not include those vertices removed because their degrees drop to zero.)
By the definition of heavy vertices, we have
$$
{\edgewt{G_{i}}-\edgewt{G_{i+1}}}
\ge d \cdot \wt{G_i}\cdot \wt{S_i}.
$$

Given this connection, we upperbound $\wt{S}=\sum_{i=0}^{s-1} \wt{S_i}$ by analyzing
  the following sum:
$$
\sum_{i=0}^{s-1}\frac{\wt{G_{i}}-\wt{G_{i+1}}}{\wt{G_{i}}}
\le 1+\sum_{i=0}^{s-2} \int_{\wt{G_{i+1}}}^{\wt{G_i}}(1/u)\hspace{0.03cm}du
=1+\int_{\wt{G_0}}^{\wt{G_{s-1}}} (1/u)\hspace{0.03cm}du = O\left(\log (n/\epsilon)\right),
$$
where the last inequality follows from $\wt{G_0}\le n$ and
  $\wt{G_{s-1}}\ge \epsilon/4$ (since
  any of its vertex cover, e.g., by taking all vertices on the left side,
  has weight at least $\epsilon/4$).
Thus,
$$
\wt{S}=\sum_{i=0}^{s-1}\wt{S_i}\le \frac{1}{d}\cdot
\sum_{i=0}^{s-1}\frac{\wt{G_{i}}-\wt{G_{i+1}}}{\wt{G_{i}}}
=o(\epsilon),
$$
by the choice of $d$ in (\ref{parameters1}).
This finishes the proof of the lemma.
\end{proof}

Note that because any of its vertex cover has weight $\Omega(\epsilon)$,
  we have $\wt{L^*}=\Omega(\epsilon)$.
Let $W=\wt{G_f^*}$.
Then we also have $W=\Omega(\epsilon)$ simply because every vertex
  in $G_f^*$ has degree at least one.
Since all vertices are light,
  we have in $G_f^*$ that $\deg(y)\le d\cdot W$ for all $y\in L^*$
  and $\inwt(j)\le d \cdot W\cdot \wt{j}$ for all $j\in R^*$. 

The bipartite graph $G_f^*$ is extremely useful for the analysis of our algorithm later.
Before that we make a short detour to sketch an informal analysis of
  the tester of Dolev and Ron \cite{DolevRon} (note that our dependency on $\epsilon$ here
  is worse than their analysis) which may help the reader better understand
  the construction so far.

First, let $R'\subseteq R^*$ be the set of vertices $j\in R^*$ such that
  $\inwt(j)\ge \wt{j}\cdot W/2$.
Then
\begin{align*}
W&=\sum_{j\in R^*} \inwt(j)
\le (W/2)\cdot \sum_{j\notin R'}\wt{j} + d\cdot W\cdot \sum_{j\in R'} \wt{j}
\le (W/2)+d\cdot W\cdot \wt{R'},
\end{align*}
which implies that $\wt{R'}=\Omega(1/d)$.
Moreover, every $S\subseteq R'$ satisfies the following nice property (below
we use $N(S)$ to denote the set of neighbors of $S$ in $G_f^*$):

\begin{lemma}\label{simple2}
In $G_f^*$, every $S\subseteq R'$ satisfies $\wt{N(S)}= \Omega\left(\wt{S}/d\right)$.
\end{lemma}
\begin{proof}
Let $W=\wt{G_f^*}$. The total edge weight between $S$ and $N(S)$ in $G_f*$ is
$$
 \sum_{j\in S} \inwt(j)\le \sum_{y\in N(S)} \deg(y)\cdot \wt{y}.
$$
Because $S\subseteq R'$, the LHS is at least
$$
\sum_{j\in S} \inwt(j)\ge (W/2)\cdot \sum_{j\in S}\wt{j}
  =(W/2)\cdot \wt{S}.
$$
Since there is no heavy vertex in $G_f^*$, the RHS is at most
$$
\sum_{y\in N(S)} \deg(y)\cdot \wt{y}
\le d\cdot W \cdot \sum_{y\in N(S)}\wt{y}=d \cdot W\cdot \wt{N(S)}.
$$
The lemma follows by combining all these inequalities.
\end{proof}

\begin{remark}
We use $G_f^*$ and $R'$ to sketch an alternative and informal analysis of the tester of Dolev and Ron
\emph{\cite{DolevRon}}
  for well-supported distributions $\calD$ (which can be extended to general distributions).
Below we assume that $\epsilon$ is a constant for convenience
  so the dependency on $\epsilon$ is worse than that of \emph{\cite{DolevRon}}.
The tester starts by sampling a set $T$ of $\tilde{O}(\sqrt{n})$ pairs from
  the sampling oracle.
It then claims victory if there are two strings $x$ and $y$ from $T$ such that $f(x)=0$,
  $f(y)=1$, and $(y,h(x))$ is an edge in $G_f$. 

Let $T_1$ denote the set of $1$-strings,
  and $T_0$ denote the set of $0$-strings from $T$.
Also let $R''\subseteq R'$ denote the  set of $j\in R'$ such that
  $h(x)=j$ for some $x\in T_0$.
Since $\calD(R')=\wt{R'}=\tilde{\Omega}(1)$,
  we have $\wt{R''}=\tilde{\Omega}(1/\sqrt{n})$
  with high probability (here $\tilde{O}(\sqrt{n})$ samples
  suffice because there are only $n$ coordinates).
When this happens,
  by Lemma \ref{simple2} we have $\wt{N(R'')}=\tilde{\Omega}(1/\sqrt{n})$ as well.
The tester then rejects if
  one of the samples in $T_1$ lies in
  $N(R'')$.
This should happen with high probability if we set the hidden
  polylogarithmic factor in the number of queries large enough.
\end{remark}

Now we return to the analysis of our algorithm
  (actually we will not use $R'$ in our analysis).
Recall that $W=\wt{G_f^*}$.
Let $L'\subseteq L^*$ denote the set of $y\in L^*$ such that
  $\deg(y)\ge W/2$ in $G_f^*$. Then similarly
\begin{align*}
W=\sum_{y\in L^*} \deg(y)\cdot \wt{y}
\le (W/2)\cdot \sum_{y\notin L'}\wt{y} + d\cdot W\cdot \sum_{y\in L'} \wt{y}
\le (W/2)+d\cdot W\cdot \wt{L'},
\end{align*}
which implies that $\wt{L'}\ge 1/(2d)$.
Our analysis of Algorithm 2 heavily relies on $G_f^*$
  and $L'\subseteq L^*$. 

We summarize below all the properties we need about $G_f^*$ and $L'$.

\begin{property}\label{summary}
Assume that $\calD$ is well-supported with respect to $f$ and
  $\dist_\calD(f,\mconj)\ge \epsilon/2$.
Then\\ $G_f^*$ $=(L^*\cup R^*,E^*)$ and $L'\subseteq L^*$ defined
  above have the following properties \emph{(}letting $W=\wt{G_f^*}$\emph{)}.
\begin{enumerate}
\item  $W=\Omega(\epsilon)$ and $\wt{L'}\ge 1/(2d)$.\vspace{-0.18cm}
\item  
  $\emph{in-wt}(j)\le d\cdot W\cdot \wt{j}$ \hspace{0.08cm}for all $j\in R^*$.
  \emph{(}We only use the fact that vertices in $R^*$ are light.\emph{)}\vspace{-0.18cm}
\item Every $y\in L'$ has $\deg(y)\ge \max\left(1,W/2\right)$.
\end{enumerate}
\end{property}

\subsection{Analysis of Algorithm 2}\label{sec:mainanalysis}

We now prove Lemma \ref{main-lemma-2}.
Let $\calD$ be a well-supported probability distribution
  with respect to $f:\{0,1\}^n\rightarrow \{0,1\}$,
  such that $f$ is $(\epsilon/2)$-far from $\mconj$ with respect to $\calD$.
Let $G_f^*=(L^*\cup R^*,E^*)$ denote
  the bipartite graph defined using $f$ and $\calD$ in the previous subsection,
  with $G_f^*$ and $L'\subseteq L^*$ satisfying Property \ref{summary}.

Here is a sketch of the proof.
We first analyze Stages 1 and 2 of Algorithm 2 in Section \ref{sec:stage1},
  where we show that if a sequence of $t$ samples $(y^1,\ldots,y^t)$
  passes Stage 1 with high probability then it can be used to produce
  many sequences of strings that get rejected in Stage 2 with high probability.
Then in Section \ref{passpass}, assuming that $(f,\calD)$ passes
  Stage 1 with high probability without loss of generality,
  we use $G_f^*$ to show that $(f,\calD)$ must get rejected in Stage 2 with high probability,
  where Property \ref{summary} plays a crucial role.

\subsubsection{Analysis of Stages 1 and 2}\label{sec:stage1}

First we assume without loss of generality that $f(\11^n)=1$; otherwise
  it is rejected at the beginning of Stage 0.
As $f$ is $(\epsilon/2)$-far from~$\mconj$, we have that both $\calD(f^{-1}(0))$ and
  $\calD(f^{-1}(1))$ are at least~$\epsilon/2$.~The former follows trivially from the fact that the all-$1$ function is in $\mconj$.
For the latter, we only need to observe that the distance between $f$ and the conjunction
  of all $n$ variables with respect to $\calD$ is at most $\calD(f^{-1}(1))$, given $f(\11^n)=1$.
{Recall that since $\calD$ is well-supported with respect to $f$,
  we can skip Stage 0 and have Stage 1 and each~iteration of Stage 2
  start by drawing $(3t/\epsilon)$ fresh samples from~the sampling oracle.
It follows directly from Chernoff bound that~Stage 1 reaches Step 1.1 with probability $1-o(1)$.
Let $\calD^1$ denote the distribution of $y\in_R\calD$ conditioning on $y\in f^{-1}(1)$.
Equivalently, we have that
  Stage 1 accepts with probability $o(1)$, and with probability $1-o(1)$ it~draws a sequence of $t$ samples $y^1,\ldots,y^t$ independently from $\calD^1$ and
  then goes through Steps 1.1 and 1.2.

The same can be said about Stage 2: Stage 2 accepts with probability $o(1)$ by Chernoff bound
  and a union~bound;
  with probability $1-o(1)$, each iteration of Stage 2
  draws a sequence of $t-1$ samples $y^1,\ldots,y^{t-1}$ from $\calD^1$
  as well as one sample $x$ from $f^{-1}(0)$, proportional~to $\calD(x)$.
Since Steps 2.1 and 2.2 use only $\alpha=h(x)$ but~not the string $x$ itself,
this inspires us to introduce $\calD^0$ as the distribution over $R$
  proportional to $\wt{j}$, $j\in R$.
Hence equivalently, each iteration of Stage 2 draws an index $\alpha$ from $\calD^0$ and
  goes through Steps 2.1 and 2.2 using $y^i$ and $\alpha$.}

We introduce some notation.
Let $\calB=(B_1,\ldots,B_t)$ be a sequence of $t$
  (not necessarily distinct) $1$-sets of $f$ (i.e.,
  $f(B_i)=1$).
We refer to $\calB$~as~a \emph{$1$-sequence of length $t$}.
Let $B=\cup_i B_i$. 
We say $\calB$ \emph{passes Stage 1 with probability $c$}
  if $B$ passes Steps 1.1 and 1.2 with
  probability $c$, without being rejected.
Similarly, we let $\calB=(B_1,\ldots,B_{t-1})$ denote a $1$-sequence~of length $t-1$,
  with $B=\cup_i B_i$.
Let $\alpha\in R$.
Then we say $(\calB,\alpha)$ \emph{fails an iteration of Stage 2} with probability $c$
  if $(\calB,\alpha)$ gets rejected in Steps
  2.1 or 2.2 with probability $c$.

We now analyze $1$-sequences $\calB=(B_1,\ldots,B_t)$ that
  pass Stage 1 with high probability.
Let
$$
B_i^*=B_i-\cup_{j\ne i} B_j,\quad\text{for each $i\in [t]$.}
$$
So $B_i^*$ contains indices that are unique to $B_i$ among all
  sets in $\calB$.
Let $I_{\calB}$ denote the set of $i\in [t]$ such that $y_i\in L'$,
  where $y_i$ is the 1-string with $\zero(y_i)=B_i$.
Intuitively, $|I_\calB|$ should be large with high probability since $\calD(L')=\wt{L'}$ is large by
  Property \ref{summary}.
We say $\calB$ is \emph{strong} if
  $|I_{\calB}|\geq t/(3d)=r/3$.
Moreover, let
  $I_{\calB}^*$ denote the set of
  $i\in I_\calB$ such that $|B_i^*|\le 6 |B|/r$.

By an averaging argument we show that if $\calB$ is strong then $|I_\calB^*|$
  is at least $r/6$.

\begin{lemma}\label{ijcount}
If $\calB$ is strong, then we have $|I_{\calB}^*|\ge r/6.$
\end{lemma}
\begin{proof}
As $\sum_{i} |B_i^*|\le|B|$, the number of $B_i$ with 
  $|B_i^*|> 6 |B|/r$ is at most $r/6$.
The lemma follows.
\end{proof}

Let $\calB=(B_1,\ldots,B_t)$ denote a strong $1$-sequence of length $t$ and
  $y_i$ denote the string with $\zero(y_i)=B_i$.
We   use it to generate input pairs $(\calB',\alpha)$ to
  Stage 2, where $\calB'$
  is a $1$-sequence of length $t-1$ and $\alpha\in R$, as follows.
For each pair $(i,\alpha)$ such that
  $i\in I_{\calB}^*$ and $\alpha\in B_i\bigcap R^*$, we say
  $\calB$ generates $(\calB',\alpha)$ via $(i,\alpha)$ if
$$
\calB'=(B_1,\ldots,B_{i-1},B_{i+1},\ldots,B_t),
$$
and we call such $(i,\alpha)$ a \emph{valid} pair.
Note that as $B_i$'s are not necessarily distinct, $\calB$
  may generate the same pair $(\calB',\alpha)$ via $(i,\alpha)$ and $(j,\alpha)$,
  $i\ne j$.
In the main technical lemma of this section, Lemma \ref{maincount} below,
  we show that if $\calB$ is strong and passes Stage 1 with high probability,
  then many $(i,\alpha)$ would lead to pairs $(\calB',\alpha)$ that fail
  Stage 2 with high probability.
Before that we make a few observations.
Recall $W=\wt{G_f^*}$.

\begin{observation}\label{obobob1}
Since $y_i\in L'$, we have $B_i\cap R^*=\deg(y_i)$ in $G_f^*$ and
  $|B_i\cap R^*|\ge \max\left(1,W/2\right)$.\\
So the total number of valid pairs $(i,\alpha)$ is bounded from below
  by both $r/6$ and $rW/12$.
\end{observation}
\begin{observation}\label{obobob2}
If a valid pair $(i,\alpha)$ satisfies $\alpha\in B_i\setminus B_i^*$
  (i.e., $\alpha$ is shared by another $B_j$ in $\calB$),
  then\\ it generates a pair $(\calB',\alpha)$ that fails Stage 2 (Step 2.1) with
  probability $1$.
\end{observation}

Now we prove Lemma \ref{maincount}. 
\begin{lemma}\label{maincount}
Assume that $\calB=(B_1,\ldots,B_t)$ is a strong $1$-sequence
  that passes Stage 1 with probability\\ at least $1/2$.
Then there are at least $\Omega(rW)$ many
  valid $(i,\alpha)$ such that the pair $(\calB',\alpha)$
  generated by $\calB$ via\\ $(i,\alpha)$ fails an iteration of
  Stage 2 with probability at least $\Omega(1)$
  (a constant that does not depend on $n$ or $\epsilon$).
\end{lemma}
\begin{proof}
For convenience, we use $I$ to denote $I_\calB^*$, with $|I|=\Omega(r)$ because
  $\calB$ is strong (Lemma \ref{ijcount}).
We let $B^*=\cup_{i\in I} B_i^*$, and let $\Gamma=B^*\cap R^*$
  (which can be empty).
We first consider two special cases on $|\Gamma|$.

\textbf{Case 1:} $|\Gamma|>|B|/t$.
Note that every $j\in \Gamma$ satisfies $f(\{j\})=0$.
This implies that $\calB$ would get rejected with probability
  $1-o(1)$ in Step 1.1, contradicting 
  the assumption that $\calB$ passes it with probability $1/2$.

\textbf{Case 2:} $|\Gamma|<rW/ 24 $.
By Observation \ref{obobob1}, the number of valid pairs  $(i,\alpha)$ is
  at least $rW/ 12 $.
In this case, however, the number of valid pairs $(i,\alpha)$ such that
  $\alpha\in B_i^*$ is at most $rW/24$.
Thus, the number of valid pairs $(i,\alpha)$ such that
  $\alpha\in B_i\setminus B_i^*$ is at least $rW/24$.
The lemma follows from Observation \ref{obobob2}.

\def\calZ{\mathcal{Z}}

In the rest of the proof we assume that
  $|B|\ge t|\Gamma|$ and $|\Gamma|=\Omega(rW)$.
They together imply that
\begin{equation}\label{hahaeq}
|B|\ge t|\Gamma|= \Omega(rtW).
\end{equation}

For $\alpha\in \Gamma$ let $s_\alpha\in [t]$ be the unique
  index with $\alpha\in B_{s_\alpha}^*$.
Now we need to do some counting.

Let $\cal Z$ denote the set of all subsets $Z\subset B$ of size $r$ such that $f(Z)=1$.
Since we assumed that $\calB$ passes Stage 1 with probability at least $1/2$, it must be the case that
$$
|\calZ|\ge \left(1-O\left(\frac{1}{s}\right)\right)\cdot {|B|\choose r}.
$$
Fixing an $\alpha\in \Gamma$ with $\alpha\in B^*_{s_{\alpha}}$,
  we are interested in
$$
\calS_\alpha=\Big\{P\cup \{\alpha\}:\text{$P$ is a subset of $B\setminus B_{s_\alpha}^*$
  of size $r-1$}\Big\}\ \ \ \text{and}\ \ \
N_\alpha=|\calS_\alpha\cap \calZ|.
$$
We would like to prove a strong lower bound for $\sum_{\alpha\in \Gamma} N_\alpha$.

To give some intuition on the connection between $N_\alpha$ and the goal, notice that
  $B\setminus B^*_{s_{\alpha}}=\cup_{i\ne s_{\alpha}}B_i$.
Let $(\calB',\alpha)$ be the pair generated from $\calB$ via $(s_\alpha,\alpha)$.
If a set $P$ of size $r-1$ is drawn from
  $\cup_{i\ne s_{\alpha}}B_i$  uniformly at random, then the probability
  of $P$ leading Step 2.2 to reject $(\calB',\alpha)$, denoted by $q_\alpha$,
  is
$$
q_{\alpha}=\frac{N_\alpha}{{|B\setminus B_{s_\alpha}^*|\choose r-1}}\ge \frac{N_\alpha}{{|B|\choose r-1}}=
\frac{N_\alpha}{{|B|\choose r}\cdot \frac{r}{|B|-r+1}} \ge \frac{N_\alpha}{{|B|\choose r}}
\cdot \frac{|B|}{2r},
$$
where the last inequality used (\ref{hahaeq}) that $|B|\gg r$.
So a strong bound for $\sum_{\alpha\in \Gamma} N_\alpha$
  may lead us to the desired claim that $q_\alpha$ is large for most $\alpha\in \Gamma$.
To bound $\sum_{\alpha\in \Gamma} N_\alpha$
and avoid double counting, let
$$
\calS_\alpha'=\Big\{P\cup \{\alpha\}:\text{$P$ is a subset of $B\setminus (B_{s_\alpha}^*\cup
  \Gamma)$ of size $r-1$}\Big\}\ \ \ \text{and}\ \ \
N_\alpha'=|\calS_\alpha'\cap \calZ|.
$$
Since $\calS_\alpha'\subseteq \calS_\alpha$ and
  $\calS_\alpha'$ are now pairwise disjoint, we have $\sum_\alpha N_\alpha\ge
  \sum_{\alpha} N_{\alpha}'$ and
$$
\sum_{\alpha\in \Gamma} N_\alpha'
  =\Big|(\cup_{\alpha\in \Gamma} \calS_\alpha')\cap \calZ\Big|
  \ge \big|\cup_{\alpha\in \Gamma} \calS_\alpha'|+|\calZ|-{|B|\choose r}\ge
  \sum_{\alpha\in \Gamma} |\calS_\alpha'| -O\left(\frac{1}{s}\right)\cdot {|B|\choose r}.
$$

On the other hand, by the definition of $I_\calB^*$ we have
  $|B_{s_\alpha}^*|\le 6 |B|/r$. We also have
  $\Gamma\le |B|/t$. Thus
\begin{equation}\label{hairy}
|\calS_\alpha'|={|B\setminus (B_{s_\alpha}^*\cup \Gamma)|\choose r-1}
\ge {|B|-(7 |B|/r)\choose r-1}=\Omega\left(\frac{r}{|B|}\cdot {|B|\choose r}\right),
\end{equation}
where details of the last inequality can be found in Appendix \ref{proof:hairy}.

Using $|\Gamma|= \Omega(rW)$ and $W=\Omega(\epsilon)$,
  $r=n^{1/3}$ and $|B|\le n$, we have
$$
\sum_{\alpha\in \Gamma} |\calS_\alpha'|=\Omega\left(\frac{r|\Gamma|}{|B|}\cdot {|B|\choose r}\right)
= \omega\left(\left(\frac{1}{s}\right)\cdot {|B|\choose r}\right).
$$
As a result, we obtain the following lower bound for $\sum_{\alpha\in \Gamma} N_\alpha$:
$$
\sum_{\alpha\in \Gamma} N_\alpha=\Omega\left(\frac{r|\Gamma|}{|B|}\cdot {|B|\choose r}\right).
$$
Combining the connection between $N_\alpha$ and
  $q_\alpha$, we have $\sum_{\alpha\in\Gamma}q_{\alpha}=\Omega(|\Gamma|)$.
Since $q_\alpha\le 1$ (it is a probability) for all $\alpha$,
  it follows easily that $q_\alpha=\Omega(1)$ for $\Omega(|\Gamma|)$ many $\alpha$'s in $\Gamma$.
For each such $\alpha$, $(s_\alpha,\alpha)$ is a valid pair
  via which $\calB$ generates a pair $(\calB',\alpha)$ that gets rejected
  by Stage 2 with probability $\Omega(1)$.

The lemma then follows from $|\Gamma|= \Omega(rW)$.
\end{proof}

\subsubsection{Finishing the Proof of Lemma \ref{main-lemma-2}}\label{passpass}

Now we combine Lemma \ref{maincount} and $G_f^*$, $L'$ to
  finish the proof of Lemma \ref{main-lemma-2}.

Assume without loss of generality that Stage 1 of Algorithm 2 either accepts
  $(f,\calD)$ or passes it down to Stage 2 with probability at least $0.9$;
  otherwise we are already done.

\def\calQ{\mathcal{Q}}
\def\calP{\mathcal{P}}

Recall that  $\calD^1$ is the distribution of $y\in_R \calD$
  conditioning on $y\in f^{-1}(1)$.
We abuse the~notation a little bit and also use $\calD^1$ to denote
  the corresponding distribution on 1-sets. 
Given a $1$-seqnence $\calB=(B_1,\ldots,B_t)$ of length $t$, we write
  $p(\calB)=\Pr_{\calD^1}[B_1]\times \cdots \times \Pr_{\calD^1}[B_t]$. 
From our discussion earlier, Stage 1 accepts $(f,\calD)$ with probability $o(1)$, and
  with probability $1-o(1)$, it runs Steps 1.1 and 1.2 on a $1$-sequence $\calB$ with
  each entry $B_i$ drawn from $\calD^1$ independently.
This implies that 
$$
\sum_{\text{$1$-seq\ }\calB} p(\calB)\cdot \Pr[\hspace{0.03cm}\calB\text{\ passes Stage 1}\hspace{0.03cm}]\ge 0.8.
$$

We focus on strong $1$-sequences. We write $S$ to denote the set of strong $1$-sequences
  and let $S'$ denote the set of strong $1$-sequences that pass Stage 1 with
  probability at least $1/2$.
Because $\calD(L')=\wt{L'}\ge 1/(2d)$ we have that Stage 1 draws a
  strong $\calB$ with probability $1-o(1)$ by Chernoff bound.
As a result, we have
$$
\sum_{\calB\in S} p(\calB)\cdot \Pr[\hspace{0.03cm}\calB\text{\ passes Stage 1}
\hspace{0.03cm}]\ge 0.8-o(1)>0.7.
$$
But the LHS is at most
$$
\sum_{\calB\in S} p(\calB)\cdot \Pr[\hspace{0.03cm}\calB\text{\ passes Stage 1}\hspace{0.03cm}]\le
(1/2)\cdot \sum_{\calB\in S\setminus S'} p(\calB)
+\sum_{\calB\in S'}p(\calB)\le (1/2)+\sum_{\calB\in S'}p(\calB)
$$
and thus, $\sum_{\calB\in S'} p(\calB)=\Omega(1)$.
The remaining proof is to use this (combined with Lemma \ref{maincount},
  $G_f^*$ and $L'$) to show that a random pair
  $(\calB',\alpha)$ gets rejected in Stage 2 with high probability.

To this end, recall that $\calD^0$ is the distribution
  over $R$ proportional to $\wt{j}$, $j\in R$.
For each pair $(\calB',\alpha)$, where $\calB'$ is
  a $1$-sequence of length $t-1$ and $\alpha\in R$,
  let $q(\calB',\alpha)=\Pr_{\calD^1}[B_1']\times \cdots\times
  \Pr_{\calD^1}[B_{t-1}']\cdot \Pr_{\calD^0}[\alpha]$.

Since Stage 2 consists of $d^*=d^2/\epsilon$ iterations,
  it suffices to show that
\begin{equation}\label{tututu}
\sum_{(\calB',\alpha)} q(\calB',\alpha) \cdot \Pr[\hspace{0.03cm}\text{$(\calB',\alpha)$
  fails an iteration of Stage 2}\hspace{0.03cm}]=\Omega(\epsilon/d^2),
\end{equation}
{as Stage 2 either accepts with~probability $o(1)$, or with probability
  $1-o(1)$ each iteration of Stage 2 draws $(\calB',\alpha)$ according to $q(\cdot)$
  and runs it through Steps 2.1 and 2.2.}

\def\gen{\mathrm{gen}}

To take advantage of Lemma \ref{maincount} we use $T$ to denote
  the set of $(\calB',\alpha)$ that is generated by a $\calB$ from
  $S'$ via a pair $(i,\alpha)$ and fails an iteration of Stage 2 with probability $\Omega(1)$
  (the same constant hidden in Lemma \ref{maincount}).
For (\ref{tututu}) it then suffices to show that
\begin{equation}\label{lululu}
\sum_{(\calB',\alpha)\in T} q(\calB',\alpha) =\Omega(\epsilon/d^2).
\end{equation}
Lemma \ref{maincount} implies that for each $\calB$ in $S'$,
  there exist $\Omega(rW)$ many valid $(i,\alpha)$ such that the pair
  generated by $\calB$ via $(i,\alpha)$ belongs to $T$ (though these
  $(\calB',\alpha)$'s are not necessarily distinct).
We use $J_\calB$ to denote these pairs of $\calB$.
We also write $(\calB^i,\alpha)$ to denote the
  pair generated by $\calB$ via $(i,\alpha)$ for convenience.

Then there is the following connection between probabilities
  $p(\calB)$ and $q(\calB^i,\alpha)$:
$$
q(\calB^i,\alpha) =\frac{p(\calB)}{\Pr_{\calD^1}[B_i]}\cdot \mathrm{Pr}_{\calD^0}[\alpha]
= p(\calB)\cdot \frac{\calD(f^{-1}(1))}{\calD(B_i)}\cdot \frac{\wt{\alpha}}{\wt{R}}
  \ge \frac{\epsilon}{2}\cdot p(\calB)\cdot \frac{\wt{\alpha}}{\calD(B_i)},
$$
where the inequality follows from $\wt{R}\le 1$ and $\calD(f^{-1}(1))\ge \epsilon/2$
  since $f$ is $(\epsilon/2)$-far from $\mconj$ with respect to $\calD$.
The only obstacle for (\ref{lululu}) is to handle the double counting.
This is where $G_f^*$ and $L'$ help.

Consider the following sum (and its connection to (\ref{lululu})):
\begin{equation}\label{main-sum}
\sum_{\calB\in S'} p(\calB)\cdot |J_\calB|.
\end{equation}
On the one hand, as $|J_\calB|=\Omega(rW)$ and
  $\sum_{\calB\in S'} p(\calB)=\Omega(1)$, the sum is $\Omega(rW)$.
On the other hand,
\begin{equation}\label{tititi}
\text{(\ref{main-sum})}=\sum_{\calB\in S'}\sum_{(i,\alpha)\in J_\calB} p(\calB)
\le \frac{2}{\epsilon}\cdot \sum_{\calB\in S'} \sum_{(i,\alpha)\in J_\calB} q(\calB^i,\alpha) \cdot \frac{\calD(B_i)}{ \wt{\alpha}}.
\end{equation}
Focusing on any fixed pair $(\calB',\alpha)$ in $T$,
  the coefficient of $q(\calB',\alpha) $ in (\ref{tititi}) is given by
\begin{equation}\label{coe}
\frac{2}{\epsilon\cdot \wt{\alpha}}\cdot \hspace{0.06cm}\sum_{\substack{\calB\in \calS',\hspace{0,03cm} (i,\alpha)\in J_\calB\\ \calB^i=\calB'}}
 {\calD(B_i)}.
\end{equation}
However, fixing an $i\in [t]$, for
  $\calB$ to generate $(\calB',\alpha)$ via $(i,\alpha)$, a necessary condition
  is $\alpha\in B_i$.
This implies that the string $y$ satisfying $\zero(y)=B_i$ must be a neighbor
  of $\alpha$ in $G_f^*$ (since $y\in L'$ by definition).
As a result it follows from Property \ref{summary} that
  the sum of (\ref{coe}) with $i$ fixed is at most $2dW/\epsilon$ (with $\wt{\alpha}$ cancelled) and thus,
  the coefficient of $q(\calB',\alpha)$ of each $(\calB',\alpha)\in T$ in (\ref{tititi})
  is $O(tdW/\epsilon)$.

Combining all these inequalities, we have
$$
\Omega(rW)=\sum_{\calB\in S'} p(\calB)\cdot |J_\calB|
\le O\left(\frac{tdW}{\epsilon}\right)\cdot \sum_{(\calB',\alpha)\in T} q(\calB',\alpha) ,
$$
and (\ref{tututu}) follows.
This finishes the proof of Lemma \ref{main-lemma-2}, and completes the analysis of Algorithm 2.

\section{Lower Bound: Proof of Theorem \ref{lowertheorem}}\label{sec:lower}

In this section, we present a lower bound of $\tilde{\Omega}(n^{1/3})$
  for the distribution-free testing of monotone conjunctions,
  and prove Theorem \ref{lowertheorem}.
Our proof is based on techniques used in the $\tilde{\Omega}(n^{1/5})$
  lower bound of Glasner and Servedio
  \cite{GlasnerServedio}, with certain careful modifications on their construction and arguments.

We start by presenting two distributions of pairs $(f,\calD)$, $\yesd$ and
  $\nod$, in Section \ref{two-distributions}, such that
\begin{enumerate}
\item Every pair $(f,\calD_f)$ in the support of $\yesd$ has $f\in \mconj$; and
\item Every pair $(g,\calD_g)$ in the support of $\nod$ has $\dist_{\calD_g}(g,\mconj)\ge 1/3$.
\end{enumerate}

Let $ q=n^{1/3}/\log^3 n$. 
Let $T$ be a deterministic (and adaptive) oracle algorithm that, upon $(f,\calD)$,
  makes no more than $q$ queries to the sampling oracle and the black-box oracle each.
(Note that even though $T$ is deterministic, each of its query to the sampling
  oracle returns a pair $(x,f(x))$ with $x$ drawn from $\calD$.)

Our main technical lemma in this section
  shows that $T$ cannot distinguish $\yesd$ and $\nod$.

\begin{lemma}\label{lowerlemma}
Let $T$ be a deterministic oracle algorithm that makes at most $q$
  queries to each oracle. Then
\begin{align*}
\left|\hspace{0.05cm}\Pr_{(f,\mathcal{D}_f)\sim\mathcal{YES}}\big[\hspace{0.03cm}
T(f,\mathcal{D}_f)\ \text{accepts}\hspace{0.03cm}\big]-\Pr_{(g,\mathcal{D}_g)\sim\mathcal{NO}}\big[
\hspace{0.03cm}T(g,\mathcal{D}_g)\text{\ accepts}\hspace{0.03cm}\big]\hspace{0.05cm}\right|\leq \frac{1}{4}.
\end{align*}
\end{lemma}

Theorem \ref{lowertheorem} then follows directly from Lemma \ref{lowerlemma}
  by Yao's minimax lemma.

\subsection{The Two Distributions $\yesd$ and $\nod$}\label{two-distributions}
We need some notation.
For strings $x,y\in\{0,1\}^n$, we use $x\wedge y\in \{0,1\}^n$ to denote
  the bitwise AND of $x$ and $y$, 
  and $x\vee y\in \{0,1\}^n$ to denote the bitwise OR of $x$ and $y$.

We use the following parameters in the definition of the two distributions:
\begin{align}
\notag h={ \frac{n^{2/3}}{2\log^2 n}},\quad r=n^{1/3}\log^2 n,\quad\ell=  {n^{2/3}} +2 ,\quad m= n^{2/3},
\quad\text{and}
\quad { s=\log^2 n}. 
\end{align}

\subsubsection{The  Distribution $\mathcal{YES}$}\label{yesdist}

A draw $(f,\mathcal{D}_f)$ from the distribution $\mathcal{YES}$ is obtained using
  the following procedure:
\begin{flushleft}\begin{enumerate}
\item Select a set $R$ of size $hr+2m=(n/2)+2n^{2/3}$ from $[n]$ uniformly at random.\vspace{-0.16cm}
\item Select a  tuple of $2m$ different indices
  $(\alpha_1, \ldots,\alpha_m,\beta_1, \ldots,\beta_m)$ from $R$
  uniformly at random.\vspace{-0.16cm}
\item Partition $R'=R\setminus\{\alpha_1, \ldots,\alpha_m,\beta_1, \ldots,\beta_m\}$ into $r$ sets
  of the same size $h$ uniformly at random.\newline
We refer to each such set as a \emph{block}. \vspace{-0.16cm}
   \item For each $i\in [m]$, select $ 2\log^2 n$ blocks uniformly at random
   (and independently for different $i$'s)
   and  let $C'_i$ be their union.
So $|C'_i|=\ell-2$. Let $C_i=C'_i\cup\{\alpha_i,\beta_i\}$ for each $i\in [m]$ and
  thus, $|C_i|=\ell$.\vspace{-0.16cm}

\item For each $i\in [m]$, select $ \log^2 n$ blocks from $C'_i$
  uniformly at random and call their union together with $\{\alpha_i\}$ to be $A_i$; let $B_i=C_i\setminus A_i$. Then $A_i$ and $B_i$ partition $C_i$ and $|A_i|=|B_i|=\ell/2$.
  \vspace{-0.16cm}

\item
We define two Boolean functions $f_1,f_2:\{0,1\}^n\rightarrow \{0,1\}$ as follows:
$$
f_1(x_1,\ldots,x_n)=\bigwedge_{j\notin R} x_j\quad \text{and}\quad
f_2(x_1,\ldots,x_n)=x_{\alpha_1}\wedge x_{\alpha_2}\wedge\cdots\wedge x_{\alpha_m}.
$$
Finally, we define $f:\{0,1\}^n\rightarrow \{0,1\}$ as $f(x)=f_1(x)\wedge f_2(x)$.\vspace{-0.13cm}

\item We define distribution $\mathcal{D}_g$ as follows.
For each $i\in [m]$, let $a^i,b^i,c^i\in\{0,1\}^n$ denote the three strings
  with $A_i=\zero(a^i)$, $B_i=\zero(b^i)$, and $C_i=\zero(c^i)$. Then we have
    $f(b^i)=1$ and $f(a^i)=f(c^i)=0$.
The probabilities of $b^i$ and
  $c^i$ in $\mathcal{D}_g$ are $2/(3m)$ and $1/(3m)$, respectively, for each $i\in [m]$.
All other strings have probability zero in $\calD_g$.
\end{enumerate}\end{flushleft}
It is clear that any pair $(f,\mathcal{D}_f)$ drawn from
  $\mathcal{YES}$ has $f\in \mconj$ as promised earlier.

\subsubsection{The Distribution $\mathcal{NO}$}\label{nodist}

A draw $(g,\mathcal{D}_g)$ from the distribution $\mathcal{NO}$ is obtained using the following procedure:
\begin{flushleft}\begin{enumerate}
\item Follow the first six steps of the procedure for $\yesd$ to
  obtain
$R,A_i,B_i,C_i,\alpha_i,\beta_i,f_1,f_2$.\vspace{-0.16cm}

\item We say a string $x\in\{0,1\}^n$ is \emph{$i$-special}, for some $i\in [m]$, if it
  satisfies both conditions:\vspace{-0.1cm}
\begin{enumerate}
\item {there are at least $3\log^2 n/4$ many blocks in $A_i$,
  each of which has (strictly)\\ more than $s$ indices $j$ in it with $x_j=0$; and
\vspace{-0.03cm}}
\item
{there are at least $3\log^2 n/4$ many blocks in $B_i$,
  each of which has at most\\ $s$ indices $j$ in it with $x_j=0$.\vspace{-0.02cm}}
\end{enumerate}

\item We use $f_2$ to define a new Boolean function $g'$.
If $f_2(x)=0$ but $x$ is $i$-special for every $i$ such that $x_{\alpha_i}=0$,
  then set $g'(x)=1$; otherwise $g'(x)=f_2(x)$.
Finally, we define $g(x)=f_1(x)\wedge g'(x)$.\vspace{-0.14cm}

\item Recall the definition of strings $a^i,b^i$ and $c^i$ from $A_i,B_i$ and $C_i$.
The probability of each of these \\ $3m$ strings
  $a^i,b^i,c^i$ is set to be $1/(3m)$ in $\calD_f$, and all other strings
  have probability zero in $\calD_f$.
\end{enumerate}\end{flushleft}

It's easy to verify that for each $(g,\mathcal{D}_g)$ drawn from the $\mathcal{NO}$ distribution, we have
$$
g(a^i)=g(b^i)=1\quad \text{but}\quad \ g(c^i)=g(a^i\wedge b^i)=0.
$$
Note that $f\in \mconj$ satisfies
  $f(x\wedge y)=f(x)\wedge f(y)$.
As a result,
  at least one of $g(a^i),g(b^i)$ or $g(c^i)$ must
  be changed in order to make $f$ a monotone conjunction.
Thus, $\dist_{\calD_g}(g,\mconj)\ge 1/3$
  as promised.

\subsubsection{The Strong Sampling Oracle}\label{strongoracle}

In the rest of the section, $(f,\calD)$ is drawn from either
  $\yesd$ or $\nod$.
While each query to the sampling oracle returns a pair $(x,f(x))$,
  $f(x)$ is redundant given the definition of $\yesd$ and $\nod$:
  $f(x)=0$ if $|\zero(x)|=\ell$ and $f(x)=1$ if $|\zero(x)|=\ell/2$.

For clarity of the proof,
  we assume that $T$ has access to a sampling oracle
  that sometimes returns extra information in addition to $x\sim\calD$.
Each time $T$ queries, the oracle draws a string $x\sim \calD$. Then
\begin{flushleft}\begin{enumerate}
\item {If $x=c^k$ for some $k\in [m]$,
  the oracle returns a pair $(C_k,\alpha_k)$.
(For the lower bound proof it is easier to work on sets instead of strings so
  we let the oracle return $C_k$ instead of $c^k$.
The extra information
  in the pair is the special variable index $\alpha_k\in C_k$.)}
\item If $x=a^k$ or $b^k$ for some $k\in [m]$ (the former happens only if $(g,\calD_g)\sim \nod$),
  the oracle returns $(\zero(x),\nil)$ (so no extra information for this case).
\end{enumerate}\end{flushleft}
We will refer to this oracle as the \emph{strong sampling oracle}.
In the rest of the section we show that Lemma~\ref{lowerlemma}
  holds even if $T$ can make $q$ queries to the strong sampling oracle
  and the black-box oracle each.

Let $T$ be such an algorithm. 
Without loss of generality, we assume that $T$ starts by making
  $q$ queries to the strong sampling oracle.
Let $Q=((D_i,\gamma_i):i\in [q])$ denote the sequence of $q$ pairs that $T$ receives~in~the sampling phase, where
  each pair $Q_i=(D_i,\gamma_i)$ has either $|D_i|=\ell/2$ and $\gamma_i=\nil$ (meaning
  that $D_i$~is~$A_{k}$~or $B_{k}$ for some $k$)
  or $|D_i|=\ell$ and $\gamma_i\in D_i$ (meaning
  that $D_i$ is $C_{k}$ and $\gamma_i$ is $\alpha_{k}$ for some $k\in [m]$).
Let $\Gamma(Q)$ denote the set of integer $\gamma_i$'s in $Q$, i.e.,
  $\alpha_k$'s revealed in $Q$,
  $S(Q)\subset [n]$ denote $\cup_{i\in [q]} D_i$, and
  $I(Q)\subseteq [q]$ denote the set of $i\in [q]$ such that $|D_i|=\ell/2$.

\subsection{Simulating $\boldsymbol{T}$ with No Access to the Black-Box Oracle}\label{defT'}

Our proof of Lemma \ref{lowerlemma} follows the high-level strategy of
  Glasner and Servedio \cite{GlasnerServedio}.
We derive~a new deterministic oracle algorithm $T'$ from $T$
  that has \emph{no access} to the black-box oracle.
We then show that such an algorithm $T'$ cannot
  distinguish the two distributions $\yesd$ and $\nod$ (Lemma \ref{noquerylemma}) but $T'$
  agrees with $T$ most of the time (Lemma \ref{yeslemma} and Lemma
  \ref{nolemma}), from which Lemma \ref{lowerlemma} follows.

Now we define $T'$ from $T$.
In addition to a sequence $Q$ of $q$ samples,
  $T'$ receives the set $R\subset [n]$ used in both procedures
  for $\yesd$ and $\nod$ for free.
Given $R$ and $Q$, $T'$ simulates $T$ on $Q$ as follows (note that $T$ is not
  given $R$ but receives only $Q$ in the sampling phase): whenever $T$ queries
   about $z\in \{0,1\}^n$,
  $T'$ does not query the black-box but passes the following bit $p(z,R,Q)$ back to $T$:
\begin{align*}
p(z,R,Q)=\begin{cases}0 & \text{if $z_i=0$ for some $i\in [n]\backslash R$
  or $i\in \Gamma(Q)$;}\\
1 & \text{otherwise.}\end{cases}
\end{align*}
So $T'$ receives $R$ and makes $q$ queries to the strong sampling oracle only.


The following lemma is the first step of our proof of Lemma \ref{lowerlemma}.

\begin{lemma}\label{noquerylemma}
Let $T^*$ be any deterministic oracle algorithm that, on a pair $(f,\mathcal{D})$
  drawn from $\yesd$ or $\nod$, receives $R$ and a sequence $Q$
  of $q$ samples but has no access to the black-box oracle.
Then  %
\begin{align*}
\notag\left|\hspace{0.05cm}\Pr_{(f,\mathcal{D}_f)\sim\mathcal{YES}}\big[\hspace{0.03cm}
T^*\ \text{accepts}\hspace{0.03cm}\big]-\Pr_{(g,\mathcal{D}_g)\sim\mathcal{NO}}\big[
\hspace{0.03cm}T^*\text{\ accepts}\hspace{0.03cm}\big]\hspace{0.05cm}\right|
=o(1).
\end{align*}
\end{lemma}
\begin{proof}
We prove a stronger statement by giving the following extra information
  to $T^*$ for free:
$$
J=\Big(\big(C_i,\{A_i,B_i\},\{\alpha_i,\beta_i\}\big):i\in [m]\Big).
$$
Note that $\{A_i,B_i\}$ is given to $T^*$ but they
  are not labelled. The same can be said about $\{\alpha_i,\beta_i\}$.
Also $R$ is revealed in $J$ as $R=\cup_i C_i$.
After $J$, $T^*$ receives a sequence of $q$ samples $Q$ 
  and now needs to either accept or reject with no other information about $(f,\calD)$.
We show that $T^*$ cannot distinguish $\yesd$ and $\nod$.

\def\calJ{\mathcal{J}}

By definition, the distribution of
  $J$ when $(f,\calD)\sim \yesd$ is the same as that when $(f,\calD)\sim \nod$,
  and we use $\calJ$ to denote the distribution of $J$.
Given a tuple $J$ drawn from $\calJ$,
  we use $\calQ_J$ to denote the distribution of the sequence of $q$-samples $Q$ conditioning on $J$
  when  $(f,\calD)\sim \yesd$, and use $\calQ_J'$ to denote
  the distribution of $Q$ conditioning on $J$ when $(f,\calD)\sim \nod$.
We show that for any fixed $J$,
\begin{equation}\label{yyy}
\left|\hspace{0.05cm}\Pr_{Q\sim\calQ_J}\big[\hspace{0.03cm}
T^*\ \text{accepts $(J,Q)$}\hspace{0.03cm}\big]-\Pr_{Q\sim \calQ_J'}\big[
\hspace{0.03cm}T^*\text{\ accepts $(J,Q)$}\hspace{0.03cm}\big]\hspace{0.05cm}\right|
=o(1).
\end{equation}
The lemma then follows because procedures for $\yesd$ and $\nod$ induce the same distribution $\calJ$ of $J$.

For (\ref{yyy}), it suffices to show that $\calQ_J$ and $\calQ_J'$ are close to each other.
For this purpose, we say a sequence $Q=((D_i,\gamma_i):i\in [q])$ has \emph{no collision} if
  no two sets $D_i$ and $D_j$ of $Q$ come from $\{A_k,B_k,C_k\}$
  with the same $k$.
On the one hand, using the birthday paradox and our choices of $q$ and $m$,
  $Q\sim \calQ_J$ has a collision with probability $o(1)$.
On the other hand, when $Q$ has no collision,
  the probability of $Q$ in $\calQ_J$ is exactly the same as
  that of $Q$ in $\calQ_J'$ (which is a product of probabilities, one for each sample
  $Q_i$ in $Q$: the probability of receiving each sample $Q_i=(D_i,\gamma_i)$ is
  $1/(6m)$ if $|D_i|=\ell$ and $1/(3m)$ if $|D_i|=\ell/2$).
(\ref{yyy}) follows, and this finishes the proof of the lemma.
\end{proof}

\subsection{Algorithms $\boldsymbol{T'}$ versus $\boldsymbol{T}$ When $(f,\calD_f)\sim \yesd$}\label{sec:analysis-lowerbound1}

Next, we show that $T'$ agrees with $T$ most of the time when $(f,\calD_f)$ is
  drawn from $\yesd$, and when $(f,\calD_f)$ is drawn from $\nod$.
We first deal with the easier case of $\yesd$.
We start with some notation.

Given a sequence of $q$-samples $Q$ in the sampling phase, we use $T_Q$ to denote
  the binary decision tree of $T$ of depth $q$ upon receiving $Q$.
So each internal node of $T_Q$ is labeled a query string $z\in \{0,1\}^n$,
  and each leaf is labeled either accept or reject.
Given $Q$, $T$ walks down the tree by making queries about $f(z)$ to the black-box oracle.
Given $R$ and $Q$, $T'$ walks down the same decision tree $T_Q$
  but does not make any query to the black-box oracle;
  instead it follows the bit $p(z,R,Q)$ for each query string $z$ in $T_Q$.

We show that the probability of $T'$ accepting a pair $(f,\calD_f)\sim\yesd$
  is very close to that of $T$.

\begin{lemma}\label{yeslemma}
Let $T$ be a deterministic oracle algorithm that makes $q$ queries to the
  strong sampling oracle and the black-box oracle each, and let
  $T'$ be the algorithm defined using $T$ as in Section \ref{defT'}.
Then
$$
\left|\Pr_{(f,\calD_f)\sim \yesd} \big[\hspace{0.03cm}\text{$T$ accepts}
  \hspace{0.03cm}\big]
-\Pr_{(f,\calD_f)\sim \yesd}\big[\hspace{0.03cm}\text{$T'$ accepts}
  \hspace{0.03cm}\big]\right|\le 0.1.
$$
\end{lemma}
\def\calR{\mathcal{R}}
\begin{proof}
Given a sequence $Q$ of $q$ samples that $T$ and $T'$ receive in the sampling phase,
  we let $\yesd_{Q}$ denote the distribution of $(f,\calD_f)$ drawn
  from $\yesd$ conditioning on $Q$.
We claim that for any $Q$,
\begin{equation}\label{ttt}
\left|\Pr_{(f,\calD_f)\sim \yesd_{Q}} \big[\hspace{0.03cm}\text{$T$ accepts}
  \hspace{0.03cm}\big]
-\Pr_{(f,\calD_f)\sim \yesd_{Q}}\big[\hspace{0.03cm}\text{$T'$ accepts}
  \hspace{0.03cm}\big]\right|\le 0.1.
\end{equation}
The lemma then follows directly. In the rest of the proof we consider a \emph{fixed}
  sequence $Q$ of samples.

We use $S=S(Q)$ to denote the union of sets in $Q$ (so $|S|\le q\hspace{0.02cm} \ell=O(n/\log^3 n)$),
  and use $t=|\Gamma(Q)|$ to denote the number of $\alpha_i$'s in $Q$.
By the definition of $\yesd$, every $\alpha_i\in S$ must appear in $Q$
  since $\calD_f$ has zero probability on strings $a^i$ (so the only possibility of having
  an $\alpha_i\in S$ is because $C_i$ is in $Q$, for which case $\alpha_i$ is also given in $Q$).
Thus, there are exactly $m-t$ many $\alpha_i$'s in $R\setminus S$
  and we use $\Delta$ to denote the set of these $\alpha_i$'s.
Let $\calR_{Q}$ denote the distribution of the set $R$, conditioning on $Q$.
Given an $R$  from $\calR_{Q}$,
  we abuse the notation and use $\yesd_{{Q}, R}$ to
  denote the distribution of $(f,\calD_f,\Delta)$, conditioning on $Q$
  and $R$.

We make a few simple but very useful observations.
First the leaf of $T_{Q}$ that $T'$ reaches only depends
  on the set $R$ it receives at the beginning; we use $w'(R)$
  to denote the leaf that $T'$ reaches.
Second, conditioning on $Q$ (and $S$), all indices $i\in [n]\setminus S$
  are symmetric and are equally likely to be in $R$.
Thus,  in $\calR_{Q}$, $R\setminus S$ is a subset of $[n]\setminus S$ of size $hr+2m-|S|$
  drawn uniformly at random.
Finally, conditioning on $Q$ and an $R$ drawn from $\calR_{Q}$,
  all indices $i\in R\setminus S$ are symmetric and equally likely to be in $\Delta$ (i.e., chosen
  as an $\alpha_i$).
In $\yesd_{Q, R}$, $\Delta$ is  a subset of $R\setminus S$ of size $m-k$ drawn uniformly at random.

Now we work on (\ref{ttt}).
Our plan is to show that, when $(f,\calD_f)\sim \yesd_{Q}$,
  most likely $T$ and $T'$ reach the same leaf of $T_{Q}$ (and then
  either both accept or reject).
We need a few definitions.

For each leaf $w$ of $T_{Q}$, we define $H_w\subseteq [n]\setminus S$ to be the set of indices
  $i\in[n]\setminus S$ such that there exists a query string
  $z$ on the path from the root to $w$ but $z_i=0$ and $w$ lies in the $1$-subtree of $z$.
By the definition of $H_w$ and the way $T'$ walks down $T_{Q}$ using $R$, a necessary condition for
  $T'$ to reach $w$ is that $H_w\subset R$.
However, conditioning on $Q$, all indices $i\in [n]\setminus S$
  are symmetric and equally likely to be in $R$ drawn from $\calR_{Q}$.
So intuitively it is unlikely for $T'$ to reach $w$ if $H_w$ is large.

Inspired by discussions above, we say a leaf $w$ of $T_{Q}$ is \emph{bad} if $|H_w|\geq 0.02\cdot n^{1/3}$;
  otherwise $w$ is a \emph{good} leaf
  (notice that whether $w$ is good or bad only depends on $Q$ (thus, $S$)
  and $T_{Q}$).
We show that, when $R$ is drawn from $\calR_{Q}$,
  the probability of $w'(R)$ being bad is $o(1)$.
To see this,
for each bad leaf $w$ of $T_Q$ we have (letting $K=(n/2)+2n^{2/3}-|S|$ be the size of $R\setminus S$
  and plugging in $|S|\le q\hspace{0.02cm} \ell=O(n/\log^3 n)$)
\begin{align*}
\Pr_{R\sim \calR_{Q}}\big[w'(R)=w\big]
&\le \Pr_{R\sim \calR_{Q}}\big[H_w\subset R\big]
=\frac{{n-|S|-|H_w|\choose K-|H_w|}}{{n-|S|\choose K}} \\[0.4ex]
&=\frac{K-|H_w|+1}{n-|S|-|H_w|+1} \times\cdots\times \frac{K}{n-|S|}<2^{-|H_w|}\le 2^{-0.02\cdot n^{1/3}}.
\end{align*}
By a union bound on the at most $2^q$ many bad leaves in $T_{Q}$
  and our choice of $q=O(n^{1/3}/\log^3 n)$
  we have 
 the probability of $T'$ reaching a bad leaf is $o(1)$, when $R\sim \calR_Q$.
This allows us to focus on good leaves.

Let $w$ be a good leaf in $T_Q$, and let
  $R$ be a set from $\calR_{Q}$ such that $w'(R)=w$
  (and thus, we must have $H_w\subset R\setminus S$).
We bound probability of $T$ not reaching $w$,
  when $(f,\calD_f,\Delta)\sim\yesd_{Q, R}$.
We claim that this happens only when $\alpha_i\in H_w$ for some $i\in [m]$
  (or equivalently, $H_w\cap \Delta$ is not empty).

We now prove this claim. Let $z$ denote the first query string along
  the path from the root to $w$ such that $f(z)\ne p(z,R,Q)$.
By the definition of $\yesd$ and $p(z,R,Q)$, $p(z,R,Q)=0$ implies $f(z)=0$.
As a result, we must have $f(z)=0$ and $p(z,R,Q)=1$.
By $p(z,R,Q)=1$, we have $\zero(z)\subseteq R$
  and $\zero(z)$ has none of the $\alpha_i$'s in $\Gamma(Q)$.
By $f(z)=0$, $\zero(z)$ must contain
  an $\alpha_i$ outside of $S$, so this $\alpha_i$ is in
  $H_w\cap \Delta$.
The latter is because $p(z,R,Q)=1$ implies that $z$ is one of the strings
  considered in the definition of $H_w$.

Using this claim, our earlier discussion on the distribution of $\Delta$
  in $\yesd_{Q, R}$ and $|H_w|<0.02\hspace{0.03cm} n^{1/3}$ as $w$ is a good leaf of $T_Q$, we have
(letting $K=(n/2)+2n^{2/3}-|S|$ be the size of $R\setminus S$)
\begin{align*}
\Pr_{(f,\calD_f,\Delta)\sim \yesd_{Q, R}} \Big[
\hspace{0.03cm}\text{$T$ does not reach $w$}\hspace{0.03cm}\Big]
&\le \Pr_{(f,\calD_f,\Delta)\sim \yesd_{Q, R}} \Big[\hspace{0.03cm}
|H_w\cap \Delta|\ne \emptyset\hspace{0.03cm}\Big]\\[1ex]
&=1-\frac{{K-|H_w|\choose m-t}}{{K\choose m-t}}\le 1-\left(1-\frac{m}{K-|H_w|+1}\right)^{|H_w|}\\[0.1ex]
&\le 1-\left(1-\frac{3m}{n}\right)^{|H_w|}\le 1-\left(1-\frac{3}{n^{1/3}}\right)^{0.02\hspace{0.02cm} n^{1/3}}\\[0.6ex]
&\approx 1-e^{-0.06}<0.07.
\end{align*}

Combining this and the fact that $T'$ reaches a bad leaf with $o(1)$ probability, we have
\begin{align*}
\Pr_{(f,\calD_f)\sim \yesd_{Q}} \big[\hspace{0.03cm}\text{$T$ and $T'$ reach different
  leaves of $T_{Q}$}
  \hspace{0.03cm}\big]\\[0.6ex]
&\hspace{-7.2cm}=\sum_{w} \hspace{0.06cm}\sum_{R:\hspace{0.03cm} w'(R)=w}
  \Pr_{(f,\calD_f,\Delta)\sim \yesd_{Q, R}} \big[\hspace{0.03cm}\text{$T$ does not reach $w$}
  \hspace{0.03cm}\big]
    \cdot \Pr_{\calR_{Q}}[\hspace{0.02cm}R\hspace{0.02cm}]
 \\[0.5ex]
&\hspace{-7,2cm}=o(1)+ \sum_{\text{good $w$}}\hspace{0.06cm} \sum_{R:\hspace{0.03cm}w'(R)=w}
  \Pr_{(f,\calD_f,\Delta)\sim \yesd_{Q, R}} \big[\hspace{0.03cm}\text{$T$ does not reach $w$}
  \hspace{0.03cm}\big]
    \cdot \Pr_{\calR_{Q}}[\hspace{0.02cm}R\hspace{0.02cm}]<0.1.
\end{align*}
This finishes the proof of (\ref{ttt}) and the lemma.
\end{proof}

\subsection{Algorithms $\boldsymbol{T'}$ versus $\boldsymbol{T}$ When $(g,\calD_g)\sim \nod$}\label{sec:analysis-lowerbound2}

We work on the more challenging case when $(g,\calD_g)\sim\nod$.
We start by introducing a condition on $Q$, 
and show that $Q$ satisfies it with probability $1-o(1)$.

\begin{definition}\label{ffff}
Given a sequence $Q=((D_i,\gamma_i):i\in [q])$ of $q$ samples from
  $(g,\calD_g)\sim \nod$,
  we use $H_i$ to denote the \emph{unique} set $C_k$ for some $k\in [m]$ that contains $D_i$.
Then we say that $Q$ is
  \emph{separated} with respect to $(g,\calD_g)$
  (since by $Q$ itself one cannot tell if it satisfies the following condition) if for each $i\in [q]$
  the number of $2\log^2 n$ blocks of $H_i$ that do not appear in any other $H_j$, $j\ne i$,
  is at least $(15/8)\log^2 n$.
\end{definition}

Here is an observation that
  inspires (part of) the definition.
Assume that algorithm $T$, given $Q$,
  suspects that $D_i$ in $Q$ is $A_k$ for some $k$
  and wants to find $\alpha_k$.
However, indices that appear in $D_i$ only, $D_i\setminus \cup_{j\ne i}\hspace{0.03cm}D_j$,
  are symmetric and are equally likely to be $\alpha_k$.
$Q$ being
  separated with respect to $(g,\calD_g)$ implies that there are many such indices in $D$.
Of course the definition of $Q$ being separated is stronger, and
  intuition behind it will  become
  clear later in the proof of Lemma \ref{nolemma}. 

We show that when $(g,\calD_g)\sim \nod$, $Q$ is
  separated with respect to $(g,\calD_g)$ with probability $1-o(1)$.

\begin{lemma}\label{lem:separated}
When $(g,\calD_g)\sim \nod$, a sequence $Q$ of $q$ samples from the
  sampling oracle is
  separated with respect to $(g,\calD_g)$ with probability $1-o(1)$.
\end{lemma}
\begin{proof}
Recall that $R'$ is the subset of $R$ with $\alpha_i$'s and $\beta_i$'s
  removed.
Fix a $R'\subset [n]$ of size $hr$ and a partition of $R'$ into
  $r$ pairwise disjoint blocks of size $h$ each.
We write $J$ to denote the tuple consists of~$R'$
  and blocks in $R'$, and $\nod_J$ to denote the distribution of $(g,\calD_g)
    \sim \nod$ conditioning on $J$.
We also write $C_i'$ to denote the set obtained from $C_i$
  after removing $\alpha_i$ and $\beta_i$.
Given $J$,  each $C_i'$ is the union of $2\log^2 n$ blocks
  drawn uniformly at random from the $r$ blocks in $R'$.

Fix an $J$. Below we show that if each $C_i'$ is the union of $2\log^2 n$ random
  blocks and
  a sequence $j_1,\ldots,j_q$ is drawn from $[m]$ uniformly and independently,
  then with probability $1-o(1)$ we have for each $i\in [q]$:
\begin{equation}\label{yuyu}
\text{the number of blocks of $C_{j_i}'$ that appear in $\cup_{k\ne i}C_{j_k}'$ is at most $\log^2 n/16$.}
\end{equation}
It follows that $Q$ has the desired properties
  when $(g,\calD_g)\sim \nod_J$ with probability $1-o(1)$, and the lemma follows.
For the rest of the proof we assume that $J$ is fixed.

We now prove the claim.
First of all  by the birthday paradox and our choices of $q$ and $m$,
  the probability of two indices $j_1,\ldots,j_q$ being the same is $o(1)$.
Suppose that no two indices in $j_1,\ldots,j_q$ are the same.
The distribution of $C_{j_1}',\ldots,C_{j_q}'$ is then the same
  as $H_1,\ldots,H_q$, where
  each $H_i$ is the union of $2\log^2 n$ blocks in $J$ drawn uniformly and independently at random.
For the latter, we show that with probability $1-o(1)$:
\begin{equation}\label{appt}
\text{for each $i\in [q]$, the number of blocks in $H_i$ that appear in $\cup_{k\ne i} H_k$
  is at most $\log^2 n/16$.}
\end{equation}
This is not really surprising:
on expectation, the number of blocks of $H_i$ that also appear in $\cup_{k\ne i}\hspace{0.05cm}
  H_k$ is
$$
(q-1)\cdot \frac{2\log^2 n\cdot 2\log^2 n}{r}=o(1).
$$
A formal proof that (\ref{appt}) happens with probability $1-o(1)$ can be found in Appendix \ref{app:tt}.
\end{proof}

We write $E$ to denote the event that a sequence $Q$
  of $q$ samples drawn from $(g,\calD_g)\sim \nod$ is separated with respect to $(g,\calD_g)$,
  and $\calQ_E$ to denote the probability distribution of $Q$ conditioning on $E$.
By definition not every $Q$ is in the support of $\calQ_E$; we record the following property
  of $Q$ in the support of $\calQ_E$.

\begin{property}\label{trivial:1}
Given any $Q=((D_i,\gamma_i):i\in [q])$ in the support of $\calQ_E$,
  each $D_i$ has at most $\log n^2/8$ many blocks that appear in $\cup_{j\ne i} D_j$.
\end{property}

Given a $Q$ in the support of $\calQ_E$, we write $\calR_{Q,E}$ to denote the
   distribution of $R$, conditioning on $Q$ and $E$.
It is clear that $\calR_{Q,E}$ is the same as $\calR_{Q}$ with $E$ dropped
  since all indices in $[n]\setminus S(Q)$ remain symmatric and equally likely to be
  in $R$ even given $E$.

\begin{property}\label{trivial:2}
\hspace{-0.02cm}For  $R \sim \calR_{Q,E}$, \hspace{-0.02cm}$R\hspace{-0.02cm}\setminus\hspace{-0.02cm} S(Q)$
  is a set of size $hr+2m-|S(Q)|$
  drawn uniformly from $[n]\hspace{-0.02cm}\setminus\hspace{-0.02cm} S(Q)$.
\end{property}
\def\calF{\mathcal{F}}
Given $Q=((D_i,\gamma_i):i\in [q])$,
   we use
   $F_i$ to denote the other set of size $\ell/2$ paired with $D_i$, $i\in I(Q)$
   (so $F_i$ is $A_k$ if $D_i$ is $B_k$ and vice versa).
Given $Q=((D_i,\gamma_i):i\in [q])$ in the support of $\calQ_E$
  and $R$ in the support of $\calR_{Q,E}$,
  we use $\calF^i_{R,Q,E}$ to denote the distribution
  of $F_i$ conditioning on $R,Q$ and $E$.
Then

\begin{property}\label{trivial:3}
Every $F_i$ in the support of $\calF^i_{R,Q,E}$
  has at least $(7/8)\log^2 n$ blocks in $R\setminus S(Q)$.
Moreover, they are drawn uniformly at random from blocks in $R\setminus S(Q)$.
(More exactly, the number $k$ of blocks of $F_i$ in $R\setminus S(Q)$ is drawn
  from a certain distribution, where $k\ge (7/8)\log^2n$ with probability $1$,
  and then $k$ blocks are drawn uniformly at random from blocks in $R\setminus S(Q)$.)
\end{property}

We now show that $T'$ agrees with $T$ most of the time
  when $(g,\calD_g)\sim \nod$:

\begin{lemma}\label{nolemma}
Let $T$ be a deterministic oracle algorithm that makes $q$ queries to the
  strong sampling oracle and the black-box oracle each, and let
  $T'$ be the algorithm defined using $T$ as in Section \ref{defT'}.
Then
$$
\left|\Pr_{(g,\calD_g)\sim \nod} \big[\hspace{0.03cm}\text{$T$ accepts}
  \hspace{0.03cm}\big]
-\Pr_{(g,\calD_g)\sim \nod}\big[\hspace{0.03cm}\text{$T'$ accepts}
  \hspace{0.03cm}\big]\right|\le 0.1.
$$
\end{lemma}

\begin{proof}
Let $Q$ be a sequence of $q$ samples in the support of $\calQ_E$.
We prove that for any such $Q$:
\begin{equation}\label{ttt2}
\left|\Pr_{(g,\calD_g)\sim \nod_{Q,E}} \big[\hspace{0.03cm}\text{$T$ accepts}
  \hspace{0.03cm}\big]
-\Pr_{(g,\calD_g)\sim \nod_{Q,E}}\big[\hspace{0.03cm}\text{$T'$ accepts}
  \hspace{0.03cm}\big]\right|\le 0.09.
\end{equation}
The lemma then follows from (\ref{ttt2}) and Lemma \ref{lem:separated}.
Below we consider a \emph{fixed} $Q$ in the support of $\calQ_E$.

For convenience, we let $S=S(Q)$, $\Gamma=\Gamma(Q)$ and $I=I(Q)$ since $Q$ is fixed
  (so $|S|=O(n/\log^3 n)$).
Given $R$ in the support of
  $\calR_{Q,E}$, we let $w'(R)$ denote the leaf of $T_{Q}$ that $T'$ reaches given $R$.
We define $H_w$ for each leaf $w$ of $T_{Q}$ and
  \emph{good}\hspace{0.05cm}/\hspace{0.05cm}\emph{bad} leaves of $T_{Q}$ similarly as in
  the proof of Lemma \ref{yeslemma}.
Using the same argument (as by Property \ref{trivial:2}, $R\setminus S$ is also drawn
  uniformly at random from $[n]\setminus S$) we have the probability of
  $w'(R)$ being bad is $o(1)$ when $R\sim \calR_{Q,E}$.
This again allows us to focus on good leaves in $T_Q$.

Now we fix a good leaf $w$ of $T_{Q}$ and  a set $R$ from
  $\calR_{Q,E}$ with $w'(R)=w$.
We use $P_w$ to denote the~path of query strings from the root to $w$.
We  drop $R$ and $Q$ in $p(z,R,Q)$ since they are fixed.
In the rest of the proof we bound the probability of $T$ not reaching $w$, when $(g,\calD_g)\sim \nod_{R,Q,E}$ (conditioning on $R,Q,E$).

We consider all the possibilities of $T$ not reaching $w$.
This happens because, for some $z$ on the path $P_w$, $p(z)\ne f(z)$.
By the definition of $\nod$, at least one of the following four events holds.
We bound the probability of~each event by $o(1)$, when $(g,\calD_g)\sim \nod_{R,Q,E}$,
  and apply a union bound.
For the four events below, Events $E_0, E_1$ and $E_2$ cover
  the case when $p(z)=1$ but $f(z)=0$. Event
  $E_3$ covers the case when $p(z)=0$ but $f(z)=1$ for some $z$ in $P_w$.
(Recall that $s=\log^2 n$.)

\begin{flushleft}\begin{enumerate}
\item[] Event $E_0$: There is a string $z$ in $P_w$ such that $p(z)=1$ (so
  $w$ is in the $1$-subtree of $z$) but
  $z_{\alpha_k}=0$\\ for some $\alpha_k\notin S$. 

\item[] Event $E_1$: There is a $z$ in $P_w$ such that $p(z)=1$
  but {1) $z_{\alpha_k}=0$ for some $\alpha_k\in S$ and $\alpha_k\notin \Gamma$; \\ 2) $z$ is not $k$-special because
  there are more than $\log^2 n/4$ many blocks in $A_k$, each of which
  has\\ at most $s$ indices $j$ with $z_j=0$.}

\item[] Event $E_2$: There is a $z$ in $P_w$ such that $p(z)=1$
  but {1) $z_{\alpha_k}=0$ for some $\alpha_k\in S$ and $\alpha_k\notin \Gamma$;\\
  2) $z$ is not $k$-special because
  there are more than $\log^2 n/4$ many blocks in $B_k$,
  each of which has \\ (strictly) more than $s$ indices $j$ such that $z_j=0$.}

\item[] Event $E_3$: There is a $z$ in $P_w$ such that $z_{\alpha_k}=0$ for some
  $\alpha_k\in \Gamma$ but $z$ is $k$-special, i.e.,
{there are\\ at least $3\log^2 n/4$ blocks in $A_k$,
  each of which has (strictly) more than $s$ indices $j$ in it with $z_j=0$;
there are at least $3\log^2 n/4$ blocks in $B_k$,
  each of which has at most $s$ indices $j$ in it with $z_j=0$.}
\end{enumerate}\end{flushleft}
The probability of $E_0$ under $\nod_{R,Q,E}$ is less than $0.07$ by the
  same argument in the proof of Lemma \ref{yeslemma}.

Next we bound the probability of $E_1$.
Let $D_i'=D_i\setminus (\cup_{j\ne i}D_j)$ for each $i\in [q]$.
Note that if there is an $\alpha_k\in S$ but
  $\alpha_k\notin \Gamma$, then $\alpha_k\in D_i'$ for some $i\in I$.
Fixing a query string $z$ in $P_w$ and an $i\in I$,
  we bound the probability that $E_1$ happens at $z$ and $\alpha_k\in D_i'$,
  and then apply a union bound on at most $q^2$ pairs of $z$ and $i$.

Consider the scenario that $D_i$ is indeed $A_k$ for some $k$;
  otherwise $E_1$ can never happen.
When $D_i$ is $A_k$, $D_i'$ consists of $\{\alpha_k\}$ and
  $u\ge 7\log^2 n/8$ blocks. (Note that $u$ can be determined from the size of $|D_i'|$.)
A key observation is that, conditioning on $R,Q$ and $E$, all indices
  in $D_i'$ are symmetric. So the choice of $\alpha_k$ as well as
  the partition of the rest of $D_i'$ into $u$ blocks are both done uniformly at random.
Let $Z=\zero(z)\cap D_i'$.
By the observation above, part 1) of $E_1$ happens with probability
  $|Z|/|D_i'|=O(|Z|/\ell)$.
So to make part 1) happen, one would like to set $Z$ to be as large as possible.
However, we claim that if $|Z|\ge 10\log^4 n$, then
  with high probability, every block in $D_i'$
  has at least $2s$ indices in $\zero(z)$, from which we know
  part 2) is violated because by $E$ the number of blocks in $D_i\setminus D_i'$ is
  at most $\log^2 n/8$.

The claim above is not surprising, since each block by our discussion earlier
  is a subset of size $h$ drawn from $D_i'$ uniformly at random.
So when $Z\ge 10\log^4 n$, the expected number of indices of a block in $Z$
  is $$|Z|\cdot \frac{h}{|D_i'|}\ge (10\log^4 n)\cdot \frac{n^{2/3}}{2\log^2 n}\cdot \frac{1}{n^{2/3}+2}
  \ge 4\log^2 n=4s.$$
For a formal proof of the claim, we assume that blocks in $D_i'$ are labelled:
  $D_i'$ is  partitioned into $\alpha_k$ and $u$ blocks uniformly at random and
  then the blocks are labelled uniformly at random from $1$ to $u$.
Focusing~on the block labelled $j$ it is a set of size $h$ drawn from $D_i'$ uniformly at random
  and thus, can be also generated as a sequence of indices drawn from $D_i'$
  uniformly at random and independently until $h$ distinct
  indices are sampled.
However, even if we draw a sequence of $h$ indices from $D_i'$ uniformly at random
  and independently the probability of
  having at least $2s$ samples in $Z$ is already $1-n^{-\Omega(\log n)}$, e.g.,
  by a folklore extension of Chernoff bound (see Lemma \ref{lerandomlemma}).
Thus, the probability of block $j$ having at most $2s$ indices in $\zero(z)$
  is bounded by $n^{-\Omega(\log n)}$.
By a union bound on all blocks in $D_i'$, we have
  that
  every block in $D_i'$ has at least $2s$ indices in $\zero(z)$ with
  probability $1-n^{-\Omega(\log n)}$.

Combining the two cases when $Z$ is small and large,
  we have that $E_1$ happens at a fixed $z$ and $D_i$ with probability $O(\log^4 n/\ell)$.
Applying a union bound, $E_1$ happens with probability
  $O(q^2 \log^4 n/\ell)=o(1)$.

Next we consider $E_2$.
Let $Q=((D_i,\gamma_i):i\in [q])$, and $F_i$ denote the set
  paired with $D_i$ for each $i\in I$.
A necessary condition for part 2) of $E_2$ to happen
  is that there exists an $i\in I$ such that more than $\log^2 n/8$ blocks of $F_i$
  outside of $S$ has more than $s$ indices in $H_w$.
To see this is the case consider a $z\in P_w$~and~$k$~such that $E_2$ happens at $z$ and $\alpha_k$.
Then it must be the case that $A_k$ is in $Q$ and
  $B_k$ is one of the $F_i$'s.
By part 2) of $E_2$, more than $\log^2 n/4$ blocks of $B_k$ has more than $s$ indices
  in $\zero(z)$.
Given $E$, we know that at least $\log^2 n/8$ many such blocks are outside of $S$,
  each of which has more than $s$ indices in $\zero(z)$.
By $p(z)=1$ $z$ is one of the strings used to define $H_w$.
Thus, all indices of $\zero(z)$ outside of $S$ belong to $H_w$.

We fix an $i\in I$ (and apply a union bound later).
Also note that $H_w$ is a fixed set in $R\setminus S$ of size at most $0.02\hspace{0.03cm}n^{1/3}$
  because $w$ is a good leaf of $T_Q$.
Consider any partition of $R\setminus S$ into blocks (and certain number~of
  $\alpha_i$'s and $\beta_i$'s).
Then by the size of $H_w$, only $O(n^{1/3}/s)$ many of them
  can have an intersection of size more than $s$ with $H_w$,
and a block drawn uniformly at random from $R\setminus S$ is one such block
  with probability only $O(1/\log^4n)$. By Property \ref{trivial:3}
  $\smash{F_i\sim \calF^i_{R,Q,E}}$ draws at most $\log^2 n$
  blocks uniformly at random from those in $R\setminus S$.
The probability that more than $\log^2 n/8$ of them have an
  intersection of size more than $s$ with $H_w$
  can be bounded by $n^{-\Omega(\log^4 n)}$ (e.g., by following a similar argument used
  in Appendix \ref{app:tt} and considering a sequence of $2\log^2 n$ blocks
  sampled uniformly and independently).
By applying a union bound on all $i\in I$ we have that $E_2$ happens with probability
  $o(1)$ when $(g,\calD_g)\sim \nod_{R,Q,E}$.

For event $E_3$ we bound the probability that $E_3$ happens for some string
  $z$ in $P_w$ and some $\alpha_k\in \Gamma$, and then apply a union bound
  on at most $q^2$ many pairs of $z$ in $P_w$ and $\alpha_k\in\Gamma$.
Consider an adversary that picks a string $z$ and aims to make $E_3$ happen
  on $z$ and $\alpha_k$ with probability as high as possible, given $R,Q$ and $E$.
Since $\alpha_k\in \Gamma$, $C_k$ is a set in $Q$ (paired with $\alpha_k$ as a sample $Q_i$).
To ease the proof, we reveal $\beta_k$ and all the blocks in $C_k$ to
  the adversary for free and denote this information by $J$.
Next, consider the distribution of $A_k$ and $B_k$ conditioning on
  $J,R,Q$ and $E$.
A key observation is that all blocks in $J$ are equally likely to be in $A_k$
  and $B_k$: $A_k$ is the union of $\alpha_k$ and $\log^2 n$ blocks drawn uniformly
  at random from $J$, and $B_k$ is the union of $\beta_k$ and the rest of blocks from $J$.
This is because, given $E$ and that $C_k$ is in $Q$,
  neither $A_k$ nor $B_k$ is in $Q$.
Thus, neither of $J,R,Q,E$ reveals any information about how blocks in $C_k$
  are partitioned. 

Let $M$ denote the set of blocks in $J$ that have more than $s$ indices in $\zero(z)$.
For event $E_3$ to happen, $A_k$ draws $\log^2 n$ blocks from $J$ uniformly at random and have
  to hit $3\log^2 n/4$ blocks in $M$, while the~rest can only have $\log^2 n/4$
  blocks in $M$, which is highly unlikely.
For a formal proof,
  note that $M$ must have at least $3\log^2 n/4$ blocks;
  otherwise the event never happens. Also, $M$ certainly has at most $2\log^2 n$ blocks.
We sample $B_k$ using the following procedure: include in the first phase
  each block in $B_k$ independently with probability $1/2$
  and then either~add  or remove random blocks to make $B_k$ with $\log^2 n$ blocks.
By Chernoff bound, we have that with probability $1-n^{-\Omega(\log n)}$
  the first phase gets a $B_k$ with at least $(11/32)\log^2 n$ blocks in $M$
  and at most $(33/32)\log^2 n$ blocks in total (since the expectation for number of blocks is between $3\log^2 n/8$ and $\log^2 n$).
When this happens, $B_k$ sampled at the end
  must have at least $(5/16)\log^2 n>\log^2 n/4$ blocks in $M$.

Applying a union bound on all $z$ in $P_w$ and $\alpha_k$ in $\Gamma$,
  we have that $E_3$ happens with probability $o(1)$.

Combining these bounds on the probability of events $E_i$, $i\in \{0,1,2,3\}$,
  we have the probability of $T$ not reaching $w$ when $(g,\calD_g)\sim \nod_{R,Q,E}$
  is less than $0.08$.
The lemma then follows.
\end{proof}

\subsection{Putting All Pieces Together}

We now combine all the lemmas to prove Lemma \ref{lowerlemma}.

\begin{proof}[Proof of Lemma \ref{lowerlemma}]
Let $T$ be a deterministic oracle algorithm that makes at most $q$ queries to
  each oracle, and $T'$ be the algorithm that simulates $T$ with no access to the black-box oracle.
By Lemmas \ref{noquerylemma}, \ref{yeslemma}, \ref{nolemma}:
\begin{align*}
\left|\hspace{0.05cm}\Pr_{(f,\mathcal{D}_f)\sim\mathcal{YES}}\big[\hspace{0.03cm}
T(f,\mathcal{D}_f)\ \text{accepts}\hspace{0.03cm}\big]-\Pr_{(g,\mathcal{D}_g)\sim\mathcal{NO}}\big[
\hspace{0.03cm}T(g,\mathcal{D}_g)\text{\ accepts}\hspace{0.03cm}\big]\hspace{0.05cm}\right|\leq
o(1)+0.1+0.1<1/4.
\end{align*}
This finishes the proof of Lemma \ref{lowerlemma} (and Theorem \ref{lowertheorem}).
\end{proof}

\section{Extending the Upper Bound to General Conjunctions}\label{upperboundgeneral}

In this section, we prove Theorem \ref{uppertheoremforgeneral} using a simple reduction
  based on the following connection between $\mconj$ and $\conj$.
We need some notation. Given any $x\in \{0,1\}^n$ and $C\subseteq [n]$, we use $x^{(C)}$ to denote
  the string obtained from $x$ by flipping all coordinates in $C$.
Given a probability distribution $\calD$ over $\{0,1\}^n$, we use
  $\calD^{(C)}$ to denote the distribution with $\calD(x)=\calD(x^{(C)})$ for all $x$.

\begin{lemma}\label{conjunctionreduction}
Let $\calD$ be a probability distribution over $\{0,1\}^n$, $f:\{0,1\}^n\rightarrow\{0,1\}$
  be a Boolean function, and $x^*\in\{0,1\}^n$ be a string such that $f(x^*)=1$.
Let $C=\zero(x^*)$, and let $g:\{0,1\}^n\rightarrow\{0,1\}$
  denote the Boolean function with $g(x)=f(x^{(C)})$ for all $x\in \{0,1\}^n$. Then we have
\begin{enumerate}
\item If $f\in \conj$, then $g\in \mconj$.\vspace{-0.1cm}

\item If $\dist_\calD(f,\conj)\ge \epsilon$, then $\dist_{\calD^{(C)}}(g,\mconj)\ge \epsilon$.
\end{enumerate}
\end{lemma}

\begin{proof}
Assume that $f\in \conj$. Then
$$
f(x)=\left(\bigwedge_{i\in S}x_i\right)\bigwedge\left(\bigwedge_{i\in S'}\overbar{x_i}\right),
$$
where $S,S'\subseteq [n]$ are disjoint (since $f(x^*)=1$).
We also have that $C\cap S=\emptyset$ and $S'\subseteq C$.
As a result, $$g(x)=f(x^{(C)})=\bigwedge_{i\in S\cup S'} x_i\in \mconj,$$ and the first part of the lemma follows.


We prove the contrapositive of the second part.
Assume that $\dist_{\calD^{(C)}}(g,h)<\epsilon$, for some $h\in \mconj$.
Let $h'$ denote the Boolean function with $h'(x)=h(x^{(C)})$.
Then we have $h'\in \conj$ and
\begin{align*}
\dist_\calD(f,\conj)\le \dist_\calD(f,h')&=\Pr_{x\in \calD} \big[f(x)\ne h'(x)\big]\\
&=\Pr_{x\in \calD} \big[ g(x^{(C)})\ne h(x^{(C)})\big]
=\Pr_{x\in \calD^{(C)}} \big[g(x)\ne h(x)\big]=\dist_{\calD^{(C)}}(g,h)<\epsilon.
\end{align*}
This finishes the proof of the second part of the lemma.
\end{proof}

Now we prove Theorem \ref{uppertheoremforgeneral}.

\begin{proof}[Proof of Theorem \ref{uppertheoremforgeneral}]
Given Lemma \ref{conjunctionreduction},
  a distribution-free testing algorithm for $\conj$ on $(f,\calD)$
  starts by drawing $O(1/\epsilon)$ samples from $\calD$ to find a string $x^*$ with $f(x^*)=1$.
If no such string is found, the algorithm accepts; otherwise the algorithm
  takes the first sample $x^*$ with $f(x^*)=1$
  and runs our algorithm~for $\mconj$ to test whether $g(x)=f(x^{(C)})$ is in $\mconj$, where
  $C=\zero(x^*)$, or $g$ is $\epsilon$-far from $\mconj$ with respect to $\calD^{(C)}$,
(Note that we can simulate queries on $g$ using
  the black box for $f$ query by query; we can simulate
  samples drawn from $\calD^{(C)}$ using the sampling oracle for $\calD$ sample by sample.)
and returns the same answer.

This algorithm is clearly one-sided given Lemma \ref{conjunctionreduction} and the fact that our
  algorithm for testing $\mconj$ is one-sided.
When $f$ is $\epsilon$-far from $\conj$, we have that $\calD(f^{-1}(1))\ge \epsilon$ because
  the all-$0$ function is in $\conj$ (when both $x_i$ and $\overbar{x_i}$
  appear in the conjunction for some $i\in [n]$).
As a result, the algorithm finds an~$x^*$~with $f(x^*)=1$ within the first $O(1/\epsilon)$ samples
  with high probability.
It then follows from Lemma~\ref{conjunctionreduction}~that~$(f,\calD)$ is rejected with high probability.
\end{proof}

\section{Extending the Lower Bound to General Conjunctions and Decision Lists}\label{lowerboundgeneral}

\def\ltf{\textsc{Ltf}}

Let $\conj,\dl$ and $\ltf$ denote the classes of all
  general conjunctions, decision lists, and linear threshold functions, respectively.
Then we have $\mconj\subset\conj\subset\dl\subset \ltf$.
In this section, we prove Theorem \ref{lowertheoremgeneral} for
  general conjunctions and decision lists.
For this purpose we follow the same strategy used in \cite{GlasnerServedio}
  and prove the following property on the distributions $\nod$
  defined in Section \ref{two-distributions}:

\begin{lemma}\label{hehelast}
With probability $1-o(1)$, $(f,\calD_f)$ drawn from $\nod$ satisfies
  $\dist_{\calD_f}(f,\dl)\ge 1/{12}$.
\end{lemma}
The same lower bound for $\conj$ and $\dl$ then follows directly from Lemma \ref{lowerlemma},
  given that $\mconj\subset \conj\subset \dl$ and the fact that any pair
  $(g,\calD_g)$ drawn from $\yesd$ satisfies $g\in \mconj$.


\begin{proof}[Proof of Lemma \ref{hehelast}]
Let $(f,\calD_f)$ be a pair drawn from $\nod$.
Given any $i,j\in [m]$ such that $C_i\cap C_j=\emptyset$,
  we follow the same argument from Glasner and Servedio \cite{GlasnerServedio} to show that
  no decision list agrees with $f$ on all of the following six strings $a^i,b^i,c^i,a^j,b^j,c^j$.

\def\first{\textsc{first}}

Assume for contradiction that a decision list $h$ of length $k$:
$$
(\ell_1,\beta_1),\ldots,(\ell_k,\beta_k),\beta_{k+1},
$$
agrees with $f$ on all six strings.
Let $\first(a)$ denote the index of the first literal $\ell_i$ in $h$ that is
  satisfied by a string $a$, or $k+1$ if no literal is satisfied by $a$.
Then we have
\begin{equation}\label{uiui}
 \min\big\{\first(a^i),\first(b^i)\big\}\le \first(c^i)\ \ \ \ \text{and}\ \ \ \
 \min\big\{\first(a^j),\first(b^j)\big\}\le \first(c^j).
\end{equation}
This is because by the definition of $a^i,b^i$ and $c^i$, any literal satisfied by
  $c^i$ is satisfied by either $a^i$ or $b^i$.
Next assume without loss of generality that
\begin{equation}\label{uiui2}
\first(a^i)=\min\big\{\first(a^i),\first(b^i),\first(a^j),\first(b^j)\big\}.
\end{equation}
By (\ref{uiui}) we have that $\first(c^i)\ge \first(a^i)$.
As $h(c^i)=f(c^i)=0$ and $h(a^i)=f(a^i)=1$, we have that $\first(c^i)\ne \first(a^i)$ and thus,
  $\first(c^i)>\first(a^i)$.
This implies that the literal $\smash{\ell_{\first(a^i)}}$ must be $x_k$ for some $k\in B_i$.
As $C_i\cap C_j=\emptyset$, we have $B_i\cap C_j=\emptyset$ and thus,
  $c^j_k=1$.
This implies that $\first(c^j)\le \first(a^i)$, and $\first(c^j)<\first(a^i)$
  because they cannot be the same given that $h(c^j)=f(c^j)=0$ and $h(a^i)=f(a^i)=0$.
However, $\first(c^j)<\first(a^i)$ contradicts with (\ref{uiui}) and (\ref{uiui2}).

As a result, when $C_i$ and $C_j$ are disjoint, one has to flip at least
  one bit of $f$ at the six strings to make it consistent with a decision list.
The lemma then follows from the fact that, with probabiilty $1-o(1)$, at least
  half of the pairs $C_{2i-1}$ and $C_{2i}$, $i\in [m/2]$, are disjoint.
\end{proof}

\section{Extending the Lower Bound to Linear Threshold Functions}\label{LTFsec}

In this section we extend our lower bound to the distribution-free testing of linear threshold
  functions (LTF for short).
We follow ideas from Glasner and Servedio \cite{GlasnerServedio} to construct a pair
  of probability distributions $\yesd^*$ and $\nod^*$ with the following properties:
\begin{flushleft}\begin{enumerate}
\item For each draw $(f,\calD_f)$ from $\yesd^*$, $f$ is a LTF;\vspace{-0.1cm}

\item For each draw $(g,\calD_g)$ from $\nod^*$, $g$ is $(1/4)$-far from LTFs with respect to $\calD_g$.
\end{enumerate}\end{flushleft}
Let $q=n^{1/3}/\log^3n$.
We follow arguments from the proof of Lemma \ref{lowerlemma} to prove the following
  lemma:

\begin{lemma}\label{ltflowerbound}
Let $T$ be a deterministic algorithm that makes at most $q$
  queries to each oracle. Then
\begin{align*}
\left|\hspace{0.05cm}\Pr_{(f,\mathcal{D}_f)\sim\mathcal{YES^*}}\big[\hspace{0.03cm}
T(f,\mathcal{D}_f)\ \text{accepts}\hspace{0.03cm}\big]-\Pr_{(g,\mathcal{D}_g)\sim\mathcal{NO^*}}\big[
\hspace{0.03cm}T(g,\mathcal{D}_g)\text{\ accepts}\hspace{0.03cm}\big]\hspace{0.05cm}\right|\leq \frac{1}{4}.
\end{align*}
\end{lemma}

Our lower bound for LTFs then follows from Yao's minimax lemma.
Below we define $\yesd^*$ and $\nod^*$ in Sections \ref{LTFyes} and \ref{LTFnod}, respectively,
  and prove Lemma \ref{ltflowerbound} in Section \ref{LTFproof}.



\subsection{The Distribution $\yesd^*$}\label{LTFyes}

Recall the following parameters from the definition of $\yesd$ and $\nod$ in Section \ref{two-distributions}:
\begin{align}
\notag \ell=  {n^{2/3}} +2 ,\quad m= n^{2/3},
\quad\text{and}
\quad { s=\log^2 n}. 
\end{align}
A draw $(f,\calD_f)$ from the distribution $\yesd^*$ is obtained using the following procedure:

\begin{flushleft}\begin{enumerate}
 \item Following the first five steps of the definition of
   $\yesd$ in Section \ref{yesdist} to obtain $R,C_i,A_i,B_i,\alpha_i,\beta_i$.
For each $i\in [m]$, let $a^i,b^i,c^i$ be the strings with $A_i=\zero(a^i)$,
  $B_i=\zero(b^i)$, $C_i=\zero(c^i)$.

\item  Define $u:\{0,1\}^n\rightarrow\mathbb{Z}$ as following:
$$u(x)=10n^2\sum\limits_{k\in[n]\backslash R}x_k+5n\sum\limits_{i\in[m]}x_{\alpha_i}-
\sum\limits_{k\in[n]}x_k.$$
Let $\theta=10n^2(n/2-2m)+5nm-(n-\ell/4)$.
\item Let $f:\{0,1\}^n\rightarrow\{0,1\}$ be the function with
  $f(x)=1$ if $u(x)\ge\theta$, and $f(x)=0$ otherwise.
The distribution $\calD_f$ is defined as follows: we put $1/4$  weight on $\11^n$, and for each $i\in[m]$, we put $1/(2m)$ weight on $b^i$ and $1/(4m)$ weight on $c^i$.
\end{enumerate}\end{flushleft}
Clearly every pair $(f,\calD_f)$ drawn from $\yesd^*$
  satisfies that $f$ is an LTF.
It is also easy to check that $$f(a^i)=f(c^i)=f(\11^n)=0\ \ \ \text{and}\ \ \ f(b^i)=1,
\ \ \ \ \ \text{for each $i\in[m]$}.$$

\subsection{The Distribution $\nod^*$}\label{LTFnod}

A draw $(g,\calD_g)$ from the distribution $\nod^*$ is obtained in the following procedure:

\begin{flushleft}\begin{enumerate}
 \item Following the definition of
   $\yesd$ in Section \ref{yesdist} to obtain $R,C_i,A_i,B_i,\alpha_i,\beta_i,c^i,a^i,b^i$. 

\item We follow the same definition of a string being \emph{$i$-special} for some $i\in [m]$
  as in Section \ref{nodist}. Let
$$J(x)=\big\{\hspace{0.01cm}i\in[m]:\text{$x$ is $i$-special}\hspace{0.02cm}\big\},\ \ \ \ \ \text{for each $x\in \{0,1\}^n$.}$$

\item Define $v:\{0,1\}^n\rightarrow\mathbb{Z}$ as following:
$$v(x)=10n^2\sum\limits_{k\in[n]\backslash R}x_k+5n\left(|J(x)|+\sum\limits_{i\in[m]\backslash J(x)}x_{\alpha_i}\right)-\sum\limits_{k\in[n]}x_k.$$
Let $\theta$ be the same threshold: $\theta=10n^2(n/2-2m)+5nm-(n-\ell/4)$.

 \item Let $g:\{0,1\}^n\rightarrow \{0,1\}$ be the function with
   $g(x)=1$ if $v(x)\ge\theta$, and $g(x)=0$ otherwise.
$\calD_g$ is defined as follows: we put $1/4$  weight on $\11^n$  and $1/(4m)$ weight on each of $a^i,b^i,c^i$,
 $i\in [m]$.
\end{enumerate}\end{flushleft}
For each $i\in[m]$, we still have $g(c^i)=g(\11^n)=0$, $g(b^i)=1$ but $g(a^i)$ is flipped to $1$
  (since $a^i$ is $i$-special). As $C_i=A_i\cup B_i$, we have that at least one of $g(a^i),g(b^i),g(c^i),g(\11^n)$ needs to be flipped to make $g$~an LTF. It follows from the definition of $\calD_g$ that $g$ is $(1/4)$-far from LTFs with respect to $\calD_g$.

\subsection{Proof of Lemma \ref{ltflowerbound}}\label{LTFproof}

We follow arguments used in the proof of
  Lemma \ref{lowerlemma} to prove lemma \ref{ltflowerbound}.

Let $T$ be any deterministic algorithm that makes $q$ queries to each of the two oracles.
We follow Section \ref{strongoracle}, and assume that $T$ has access to the
  following \emph{strong sampling oracle}:
\begin{flushleft}\begin{enumerate}
\item When the sampling oracle returns $c^i$ for some $i\in [m]$, it returns
  the special index $\alpha_i$ as well;\vspace{-0.1cm}
\item {For convenience we also assume without loss of generality that the
  oracle always returns a sample drawn from the marginal distribution of $\calD$
  within $\{a^i,b^i,c^i\}$ since samples of $\11^n$ are not useful in distinguishing $\yesd^*$ and $\nod^*$.}
\end{enumerate}   \end{flushleft}
We show that Lemma \ref{ltflowerbound} holds even~if~$T$ receives $q$ samples from the
  strong sampling oracle and makes~$q$ queries to the black-box oracle.
We follow the same notation introduced in Section \ref{strongoracle}.
Given a sequence $\smash{Q=((D_i,\gamma_i):i\in [q])}$~of samples that $T$ receives from
  the strong sampling oracle, let $\Gamma(Q)$ denote the set of integer $\gamma_i$'s in $Q$,
  let $S(Q)=\cup_{i\in [q]} D_i$, and let $I(Q)$ denote the set of $i\in [q]$ with $|D_i|=\ell/2$.

Next we follow Section \ref{defT'} to derive from $T$ a
  new deterministic oracle algorithm $T'$ that
  has \emph{no access} to the black-box oracle but receives $R$ in addition to the sequence of samples $Q$
  at the beginning.
We show that $T'$ cannot distinguish the two
  distributions $\yesd^*$ and $\nod^*$ (Lemma \ref{ltfdistinguish}),
  but $T'$ agrees with $T$ most of the time (Lemma \ref{ltfyescase} and Lemma \ref{ltfnocase}), from which Lemma \ref{ltflowerbound} follows.



The new algorithm $T'$ works as follows:
\begin{flushleft}\begin{enumerate}
 \item[] Given $R$ and $Q$, $T'$ simulates $T$ on $Q$ as follows
 (note that $T$ is not given $R$ but receives only $Q$ in the sampling phase):
 whenever $T$ queries about $x\in \{0, 1\}^n$, $T'$ does not
 query the oracle but computes
$$\phi(x)=10n^2\sum\limits_{k\in[n]\backslash R}x_k+5n\left(m-|I'(x)|\right)-\sum_{k\in [n]} x_k,$$
where $I'(x)=\zero(x)\cap \Gamma(Q)$, i.e., the set of all $\alpha_i$'s in
  $\Gamma(Q)$ revealed in the sampling phase such that $x_{\alpha_i}=0$.
$T'$ then passes $1$ back to $T$ if $\phi(x)\ge\theta$ and $0$ otherwise to continue the simulation of $T$.
At the end of the simulation, $T'$ returns the same answer as $T$.
\end{enumerate}\end{flushleft}

Now we are ready to prove the three lemmas mentioned above.

 The first lemma is to show that a deterministic oracle algorithm with no access to the black-box oracle cannot distinguish $\yesd^*$ and $\nod^*$ distributions with high probability.

\begin{lemma}\label{ltfdistinguish}\begin{flushleft}
Let $T^*$ be any deterministic oracle algorithm that, on a pair $(f,\mathcal{D})$
  drawn from either $\yesd^*$ or $\nod^*$, receives $R$ and a sequence $Q$
  of $q$ samples but has no access to the black-box oracle.
Then  %
\begin{align*}
\notag\left|\hspace{0.05cm}\Pr_{(f,\mathcal{D}_f)\sim\mathcal{YES}}\big[\hspace{0.03cm}
T^*\ \text{accepts}\hspace{0.03cm}\big]-\Pr_{(g,\mathcal{D}_g)\sim\mathcal{NO}}\big[
\hspace{0.03cm}T^*\text{\ accepts}\hspace{0.03cm}\big]\hspace{0.05cm}\right|
=o(1).
\end{align*}
\end{flushleft}\end{lemma}

\begin{proof}
The proof of the lemma is essentially the same as the proof of Lemma \ref{noquerylemma}.
The only difference here is that the distribution $\calD$ is also supported on $\11^n$.
But because $\calD_f(\11^n)=\calD_g(\11^n)=1/4$ in both
  $\yesd^*$ and $\nod^*$, the same proof works here.
\end{proof}

Next we show that $T'$ agrees with $T$ most of the time when $(f,\calD_f)\sim\yesd^*$:

\begin{lemma}\label{ltfyescase}\begin{flushleft}
Let $T$ be a deterministic oracle algorithm that makes $q$ queries to the
  strong sampling oracle and the black-box oracle each, and let
  $T'$ be the algorithm defined using $T$ as above.
Then
$$
\left|\Pr_{(f,\calD_f)\sim \yesd^*} \big[\hspace{0.03cm}\text{$T$ accepts}
  \hspace{0.03cm}\big]
-\Pr_{(f,\calD_f)\sim \yesd^*}\big[\hspace{0.03cm}\text{$T'$ accepts}
  \hspace{0.03cm}\big]\right|\le 0.1.
$$
\end{flushleft}\end{lemma}

\begin{proof}
Fix a sequence $Q$ of $q$ samples.
We prove the same statement conditioning on $Q$.
Let $\calR_Q$ denote the distribution of the set $R$, conditioning on $Q$.
We let $T_Q$ denote the binary decision tree of $T$ of depth $q$ upon receiving $Q$,
  and let $w'(R)$ denote the leaf that $T'$ reaches given $R$.

Following the same definition and argument used in the proof of Lemma \ref{yeslemma} (as
  $\phi(x)<\theta$ if one of the variables outside of $R$ is set to $0$), it suffices to
  show for every $R$ in the support of $\calR_Q$ such that
  $w=w'(R)$ is a \emph{good} leaf (see the definition in the proof of Lemma \ref{yeslemma}), we have that $T$ reaches $w$ with
  high probability (conditioning on both $Q$ and $R$).
Note that
$u(x)$ in the $\yesd^*$ distribution can also be written as:
$$u(x)=10n^2\sum\limits_{k\in[n]\backslash R}x_k+5n\left(m-|I(x)|\right)-\sum_{k\in [n]} x_k,$$
where $I(x)$ here is the set of all $\alpha_i$'s, $i\in [m]$, such that
  $x_{\alpha_i}=0$.
Since $\phi(x)\ge u(x)$,
  $T$ does not reach $w$ if and only if one of the strings $x$ along the path
  from the root of $T_Q$ to $w$
  satisfies $$|I'(x)|<|I(x)|\ \ \ \text{and}\ \ \ \phi(x)\ge \theta>u(x).$$
Given that $\Gamma(Q)$ contains all $\alpha_i$'s in $S(Q)$ (as $a^i$ is not in the support of $\calD_f$)
  it must be the case that $x_{\alpha_i}=0$ for some $\alpha_i\notin S(Q)$ and thus,
  $\alpha_i\in H_w$ for some $i\in [m]$ (see the definition of $H_w$ in the proof of Lemma \ref{yeslemma}).
This is exactly the same event analyzed in the proof of Lemma \ref{yeslemma}, with its probability bounded from
  above by $0.1$. This finishes the proof of the lemma.
\end{proof}

Finally we show that $T'$ agrees with $T$ most of the time when $(g,\calD_g)\sim\nod^*$:

\begin{lemma}\label{ltfnocase}\begin{flushleft}
Let $T$ be a deterministic oracle algorithm that makes $q$ queries to the
  strong sampling oracle and the black-box oracle each, and let
  $T'$ be the algorithm defined using $T$ as above.
Then
$$
\left|\Pr_{(g,\calD_g)\sim \nod^*} \big[\hspace{0.03cm}\text{$T$ accepts}
  \hspace{0.03cm}\big]
-\Pr_{(g,\calD_g)\sim \nod^*}\big[\hspace{0.03cm}\text{$T'$ accepts}
  \hspace{0.03cm}\big]\right|\le 0.1.
$$
\end{flushleft}\end{lemma}

\begin{proof}
Following Definition \ref{ffff} and Lemma \ref{lem:separated},
  the event $E$ of $Q$ being \emph{separated} (with respect to $(g,\calD_g)$)
  happens with probability $1-o(1)$.
Let $\calQ_E$ denote the probability distribution of $Q$ conditioning on $E$.
Fix~a sequence $Q$ in the support of $\calQ_E$.
Below we prove the statement of the lemma conditioning on both
  $Q$ and $E$.
Let $\calR_{Q,E}$ denote the distribution of $R$ conditioning on $Q$ and $E$.

Similar to the proof of Lemma \ref{nolemma}, it suffices to show that for every $R$
  in the support of $\calR_{Q,E}$ such that $w=w'(R)$ is a good leaf,
  $T$ reaches $w$ with high probability, conditioning on $R,Q$ and $E$.


Note that $v(x)$ from the $\nod^*$ distribution can be also written as:
$$v(x)=10n^2\sum\limits_{k\in[n]\backslash R}x_k+5n\left(m-|I(x)|\right)-\sum_{k\in [n]} x_k,$$
where $I(x)$ is the set of all $\alpha_i$'s, $i\in [m]$,
  such that $x_{\alpha_i}=0$ and $x$ is not $i$-special.
Then $T$ does not reach $w$ only if for some $x$ along the path from the root of $T_Q$ to $w$,
  either $\phi(x)\ge \theta>v(x)$ or $v(x)\ge \theta>\phi(x)$.

When $\phi(x)\ge \theta>v(x)$, we have $|I(x)|>|I'(x)|$ and thus,
  one of the following two events must hold:
\begin{enumerate}
\item[] Event $E_0^*$: $\phi(x)\ge \theta$ (so $w$ is in the $1$-subtree of $x$) and
  $x_{\alpha_k}=0$ for some $\alpha_k\notin S(Q)$;\vspace{-0.1cm}
\item[] Event $E_{1,2}^*$: $\phi(x)\ge \theta$, $x_{\alpha_k}=0$ for some
  $\alpha_k\in S(Q)$ but $\alpha_k\notin \Gamma(Q)$, and $x$ is not $k$-special.
\end{enumerate}
For the case when $v(x)\ge \theta>\phi(x)$, we have $|I'(x)|>|I(x)|$ and thus,
  the following event must hold:
\begin{enumerate}
\item[] Event $E_3^*$: $x_{\alpha_k}=0$ for some $\alpha_k\in \Gamma(x)$ and $x$ is
  $k$-special.
\end{enumerate}
Note that $E_0^*$ is the same event as $E_0$, $E_{1,2}^*$ is the same event as the union
  of $E_1$ and $E_2$, and $E_3^*$ is the same event as $E_3$ in the proof of Lemma \ref{nolemma}.
The lemma follows from bounds on their probabilities given in the proof of Lemma \ref{nolemma}.
\end{proof}

Lemma \ref{ltflowerbound} then follows from Lemmas \ref{ltfdistinguish}, \ref{ltfyescase}, and \ref{ltfnocase}. 

\begin{flushleft}
\bibliographystyle{amsalpha}
\bibliography{references}
\end{flushleft}
\appendix

\section{Proof of Inequality (\ref{hairy})}\label{proof:hairy}
We prove the last step of (\ref{hairy}).
Let $k=|B|=\Omega(rtW)\gg r$ since $W=\Omega(\epsilon)$. Let $\delta=7/r$. Then
\begin{align*}
\frac{{{k}\choose{r}}}{{{k-\delta k}\choose{r-1}}}\notag&=\frac{1}{r}\cdot
\frac{(k-\delta k+1)(k-\delta k+2)\cdots k}{(k-\delta k-r+2)(k-\delta k-r+3)\cdots (k-r)}\\[0.8ex]
\notag&=\frac{k}{r}\cdot \frac{k-\delta k+1}{k-\delta k-r+2}\cdot \frac{k-\delta k+2}{k-\delta k-r+3}
 \cdots \frac{k-1}{k-r}\\[1.8ex]
\notag&\leq\frac{k}{r}\cdot \left(\frac{k-\delta k+1}{k-\delta k-r+2}\right)^{\delta k-1}
\notag \leq\frac{k}{r}\cdot \left(1+\frac{2r}{k}\right)^{\delta k} 
 =O\left(\frac{k}{r}\right).
\end{align*}

\section{Proof of (\ref{appt})}\label{app:tt}

We use the following folklore extension of the standard Chernoff bound:

\begin{lemma}\label{lerandomlemma}
Let $p\in [0,1]$ and $X_1,\ldots,X_n$ be a sequence of (not necessarily independent)
  $\{0,1\}$-valued random variables.
Let $X=\sum_{i\in [n]} X_i$.
If for any $i\in [n]$ and any $b_1,\ldots,b_{i-1}\in \{0,1\}$:
$$
\Pr\big[\hspace{0.03cm}X_i=1\ |\ X_1=b_1,\cdots,X_{i-1}=b_{i-1}\hspace{0.03cm}\big]\le p,
$$
then we have $\Pr[X\ge (1+\delta)\cdot pn]\le e^{-\delta^2pn/3}$.
\end{lemma}
\def\calA{\mathcal{A}}
Now we prove (\ref{appt}).
Fix an $i\in [q]$ and the $2\log^2 n$ blocks in $H_i$.
Then we sample all other $q-1$ many $H_j$'s and bound the probability that
  (\ref{appt}) does not happen for $i$.
We use the following procedure to sample $H_j$'s:
for each $j\ne i$ sample a sequence of $4\log^2 n$ blocks uniformly at random
   with replacement and~set $H_j$ to be the union of the first $2\log^2 n$ distinct
  blocks sampled.
This procedure, denoted by $\calA$,
  fails if for some $j$, there are less than $2\log^2 n$ distinct blocks
  from the $4\log^2 n$ samples.
When it succeeds, $\calA$ yields the desired uniform and independent distribution.
We claim that $\calA$ succeeds with probability $1-e^{-\Omega(r)}$.

To see this, for each $j$, its $k$th sample is the same as one of the previous $k-1$ samples
  with probability at most $(k-1)/r\le 4\log^2 n/r$, no matter what the outcomes of
  the first $k-1$ samples are.
By Lemma~\ref{lerandomlemma},~$\calA$ failed at $H_j$ with probability
  $e^{-\Omega(r)}$
  because this happens only if more than $2\log^2 n$ samples have appeared before.
By a union bound on $j$, $\calA$ succeeds with probability $1-e^{-\Omega(r)}$.

Let $U$ denote the union of all $(q-1)\cdot (4\log^2 n)$ blocks sampled
  by $\calA$.
Then
\begin{align*}
\Pr\big[\hspace{0.03cm}(\ref{appt})\ \text{does not hold for $i$}\hspace{0.03cm}\big]
  &\le {\Pr\big[\hspace{0.03cm}\text{$U$ has $> \log^2 n/16$ blocks
  of $H_i$}\ |\ \text{$\calA$ succeeds}\hspace{0.04cm}\big]} \\[0.6ex]
 &\le \frac{\Pr\big[\hspace{0.03cm}\text{$U$ has $> \log^2 n/16$ blocks
   of $H_i$}\hspace{0.03cm}\big]}{\Pr\big[\hspace{0.03cm}\text{$\calA$
   succeeds}\hspace{0.04cm}\big]}.
\end{align*}
Using Chernoff bound, the probability of $U$ having  more than $\log^2 n/16$ blocks
  of $H_i$ is at most $n^{-\Omega(\log n)}$.
(\ref{appt}) follows from $\Pr[\hspace{0.03cm}\text{$\calA$ succeeds}\hspace{0.04cm}]\ge 1-e^{-\Omega(r)}$
  and a union bound on $i\in [q]$.

\end{document}